\def\isarxivversion{1}

\ifdefined\isarxivversion
\documentclass[11pt]{article}
\else
\documentclass[anon,12pt]{colt2021}
\fi

\ifdefined\isarxivversion
\usepackage{amsfonts}
\fi
\usepackage{algorithm}
\usepackage{array}
\usepackage{multirow}
\usepackage{color}
\usepackage[english]{babel}
\usepackage{graphicx}
\usepackage{ amssymb }
\usepackage{wrapfig}%,epsfig}
\usepackage{epstopdf}
\usepackage{url}
\usepackage{graphicx}
\usepackage{color}
\usepackage{epstopdf}
\usepackage{algpseudocode}
\usepackage{scrextend}
\usepackage[T1]{fontenc}
\usepackage{bbm}
\usepackage{comment}
\usepackage{tikz}
\let\C\relax
\usepackage{tikz}
\usepackage{hyperref}
\hypersetup{colorlinks=true,citecolor=blue,linkcolor=red}

\usetikzlibrary{arrows}

\graphicspath{{./figs/}}

\ifdefined\isarxivversion
\usepackage[margin=1in]{geometry}
\else

\fi

\definecolor{b2}{RGB}{51,153,255}
\definecolor{mygreen}{RGB}{80,180,0}
\definecolor{trz}{HTML}{4169e1}
\definecolor{shunhua}{RGB}{0,180,80}
\definecolor{mycrimson}{RGB}{165,28,48}

\newcommand{\Zhao}[1]{{\color{b2}[Zhao: #1]}}

\ifdefined\isarxivversion
\usepackage{amsmath}
\usepackage{amsthm}
\usepackage{subfig}

\newtheorem{theorem}{Theorem}[section]
\newtheorem{lemma}[theorem]{Lemma}
\newtheorem{definition}[theorem]{Definition}

\newtheorem{remark}[theorem]{Remark}

\else
\fi

\newtheorem{fact}[theorem]{Fact}

\newcommand{\wt}{\widetilde}
\newcommand{\ov}{\overline}

\newcommand{\N}{\mathcal{N}}
\newcommand{\R}{\mathbb{R}}

\newcommand{\rank}{\mathrm{rank}}

\renewcommand{\d}{\mathrm{d}}
\newcommand{\poly}{\mathrm{poly}}
\newcommand{\tr}{\mathrm{tr}}
\newcommand{\new}{\mathrm{new}}

\newcommand{\Tmat}{{\cal T}_\mathrm{mat}}
\newcommand{\vect}{\mathrm{vec}}

\newcommand*{\bra}[1]{\langle#1|}
\newcommand*{\ket}[1]{|#1\rangle}

\makeatletter
\newcommand*{\RN}[1]{\expandafter\@slowromancap\romannumeral #1@}
\makeatother

\begin{document}

\ifdefined\isarxivversion
\title{A Faster Quantum Algorithm for Semidefinite Programming via Robust IPM Framework}
\else

\fi

\ifdefined\isarxivversion
\author{
	Baihe Huang\thanks{\texttt{baihehuang@pku.edu.cn}. Peking University.}
	\and
	Shunhua Jiang\thanks{\texttt{sj3005@columbia.edu}. Columbia University.}
	\and
	Zhao Song\thanks{\texttt{zsong@adobe.com}. Adobe Research.}
	\and
	Runzhou Tao\thanks{\texttt{runzhou.tao@columbia.edu}. Columbia University.}
	\and
	Ruizhe Zhang\thanks{\texttt{ruizhe@utexas.edu}. The University of Texas at Austin.}
}

\date{}
\else

\fi

\ifdefined\isarxivversion
 \begin{titlepage}
     \maketitle
     \begin{abstract}

This paper studies a fundamental problem in convex optimization, which is to solve semidefinite programming (SDP) with high accuracy.  
This paper follows from existing robust SDP-based interior point method analysis due to [Huang, Jiang, Song, Tao and Zhang, FOCS 2022]. While, the previous work only provides an efficient implementation in the classical setting. This work provides a novel quantum implementation.

We give a quantum second-order algorithm with high-accuracy in both the optimality and the feasibility of its output, and its running time depending on $\log(1/\epsilon)$ on well-conditioned instances. Due to the limitation of quantum itself or first-order method, all the existing quantum SDP solvers either have polynomial error dependence or low-accuracy in the feasibility.

     \end{abstract}
     \thispagestyle{empty}
 \end{titlepage}
 \else

 \maketitle
\begin{abstract}

This paper studies a fundamental problem in convex optimization, which is to solve semidefinite programming (SDP) with high accuracy.  
This paper follows from existing robust SDP-based interior point method analysis due to [Huang, Jiang, Song, Tao and Zhang, FOCS 2022]. While, the previous work only provides an efficient implementation in the classical setting. This work provides a novel quantum implementation.

We give a quantum second-order algorithm with high-accuracy in both the optimality and the feasibility of its output, and its running time depending on $\log(1/\epsilon)$ on well-conditioned instances. Due to the limitation of quantum itself or first-order method, all the existing quantum SDP solvers either have polynomial error dependence or low-accuracy in the feasibility.

\end{abstract}

 \fi

 %\iffalse
\newpage

\pagenumbering{roman}

{\hypersetup{linkcolor=black}
\tableofcontents
}
\newpage

\pagenumbering{arabic}
\setcounter{page}{1}
%\fi

\section{Introduction}
%In semidefinite programming (SDP) we optimize a linear objective function over the intersection of the positive semi-definite (PSD) cone with an affine space. SDP is of great interest both in theory and in practice. % Many problems in operations research, machine learning, and theoretical computer science can be modeled or approximated as semidefinite programming problems. In machine learning, SDP has applications in adversarial machine learning \cite{rsl18}, learning structured distribution \cite{clm20}, sparse PCA \cite{aw08,dejl07}, robust learning \cite{dkklms16,dhl19,jlt20}, sparse matrix factorization \cite{csz20}. In theoretical computer science, SDP has been used in approximation algorithms for max-cut \cite{gw94}, coloring $3$-colorable graphs \cite{kms94}, and sparsest cut \cite{arv09}, quantum complexity theory \cite{jjuw11}, robust learning and estimation \cite{cg18,cdg19,cdgw19}, graph sparsification \cite{ls17}, algorithmic discrepancy and rounding \cite{bdg16,bg17,b19}, sum of squares optimization \cite{bs16,fkp19}. 

Mathematically, semi-definite programming (SDP) can be formulated as:
\begin{definition}[Semidefinite programming] \label{def:sdp_primal}
Suppose there are $m+1$ symmetric 
  matrices $C \in \R^{n \times n}$ and $ A_1, \cdots, A_m \in \mathbb{R}^{n \times n}$ and a vector $b \in \R^m$, the goal is to solve the following optimization problem:
\begin{align}
\label{eq:sdp_primal}
\max_{X \in \R^{n \times n}} ~ \langle C, X \rangle %\notag \\
\textup{ subject to } ~ \langle A_i, X \rangle = b_i, ~~\forall i \in [m], %\\
~ X \succeq 0,% 
\end{align}
where $\langle A,B\rangle$ denote the inner product bween the matrix $A$ and matrix $B$. 
\end{definition}

From the above formulation, the input size of an SDP instance is $m n^2$. The reason is, we have $m$ constraint matrices and each matrix has size $n \times n$. The linear programming (LP) is a simpler case than SDP, where $X \succeq 0$ and $C, A_1, \cdots, A_m$ are restricted to be $n \times n$ diagonal matrices. 
Since \cite{d47}, there is a long line of work \cite{d47,k80,k84,r88,v89_lp,ls14,ls15,cls19,lsz19,b20,blss20,sy21,jswz21,y21,dly21} about speeding up the time complexity of linear programming.  Semi-definite programming can be viewed as a more general optimization problem compared to linear programming. Semi-definite programming is a more challenging problem due to a variety of complications arising from generalizing vectors to PSD matrices. After a long line of research (\cite{s77, yn76, k80, kte88, nn89, v89, nn92, nn94, a00, km03, lsw15, jlsw20, jklps20, hjstz22})\footnote{For more detailed summary, we refer the readers to Table 1.1 and 1.2 in \cite{jklps20}, and Table 1 in \cite{hjstz22}.}, the-state-of-the-art semidefinite programming algorithms are 

\begin{itemize}
    \item Jiang, Kathuria, Lee, Padmanabhan and Song \cite{jklps20}'s algorithm 
    \begin{itemize}
        \item It runs in $O(\sqrt{n} (mn^2 + m^{\omega} + n^{\omega}))$ time, for $m = \Omega(n)$.
    \end{itemize}
    \item Huang, Jiang, Song, Tao and Zhang \cite{hjstz22}'s algorithm 
    \begin{itemize} 
        \item It runs in $O(m^{\omega} + m^{2+1/4})$ time, for $m = \Omega(n^2)$.
    \end{itemize}
\end{itemize}

Inspired by the quantum linear algebra results such as recommendation system \cite{kp17,t19,cgl+20}, linear regression \cite{gst22,cch+22}, principal component analysis \cite{t21}, non-convex optimization \cite{syz21}, in this work, we study the quantum algorithm for the semi-definite programming.

It is well-known that second-order optimization method usually achieves $\log(1/\epsilon)$ dependence in running time, while the first-order optimization method has to pay $1/\epsilon$.\footnote{For example, see \cite{ak07,gh16,al17,cdst19,lp19,ytfuc19,jy11,alo16}.} In the quantum setting, for many linear algebra and optimization tasks (e.g., solving linear systems), there exist efficient  algorithms taking $\log(1/\epsilon)$ time and outputting a quantum state encoding the solution. However, extracting  classical information from the quantum state usually needs to pay a $\poly(1/\epsilon)$ factor.  Therefore, even implementing second-order algorithm in quantum, it is unclear how to obtain $\log(1/\epsilon)$ dependence in final running time. Thus, one fundamental question is
\begin{center}
{\it
    Is there a quantum second-order SDP algorithm that has logarithmic dependence on $\epsilon$?
}
\end{center}
%In this work, we provide a positive answer for this question.

In this work, we propose a new approach towards answering this question.
More specifically, for the $m=\Omega(n^2)$ case, we design a  quantum SDP solver running in time $m^{1.75}\cdot \poly(\kappa, \log(1/\epsilon))$ and outputting a classical solution, which is faster than the existing classical SDP algorithms. In addition, the time complexity of our algorithm matches the quantum second-order SDP algorithm \cite{kp20,kps21} and has better error dependence than all the existing quantum SDP algorithms on the well-conditioned instances. We also overcome the infeasibility issue in previous quantum algorithms.

%\subsection*{Our results} %\label{subsec:results}
%We present simplified versions of our main results in the following theorems. The formal versions can be found in  Theorem~\ref{thm:quantum_main}.

We state an informal version of the main result as follows, and delay the formal version into Theorem~\ref{thm:quantum_main}.
\begin{theorem}[Our result]\label{thm:quantum_intro}
For any semidefinite programming with $m$ constraints and variable size $n \times n$, for any accuracy parameter $\epsilon$, there is a quantum algorithm that outputs a classical solution in time
$(mn^{1.5} + n^3) \cdot  \poly(\kappa, \log(mn/\epsilon))$.
\end{theorem}
The $\kappa$ includes the condition numbers of the intermediate matrices for a clean presentation. Note that almost all the previous quantum SDP algorithms also depend on these parameters. We also note that a large family of SDP instances exists with $\kappa=O(1)$. On these instances, our algorithm only pays $\log(1/\epsilon)$ in the running time, while the previous quantum algorithms still need to pay $1/\epsilon$.

\section{Quantum barrier with existing algorithms}\label{sec:quantum_barrier}
 
The existing quantum SDP solving algorithms can be classified into two kinds: quantum first-order methods \cite{bs17,aggw17,bkllsw19,ag18} based on the Arora-Kale framework \cite{ak07}, and quantum second-order method \cite{kp20,antz21,kps21} based on the primal-dual central path framework. Although these quantum algorithms have achieved quantum speedup over the classical SDP solvers in time complexity with respect to the parameters $m$ and $n$, there are still some limitations in the existing quantum SDP solving algorithms, which we discuss in below.
\paragraph{Error dependence: polynomial or logarithmic in $\epsilon$?} The approximation error $\epsilon$ is an essential parameter for SDP solving algorithms.
%For some real applications of SDP solving, the approximation error $\epsilon$ may be very small (e.g., $1/\poly(n)$). Then, the error dependence becomes essential for SDP solver. 
For the quantum first-order algorithms, they are based on the classical first-order method, which requires $\Omega(1/\epsilon)$ iterations \cite{b15,cdhs19}. %\Ruizhe{Is there any classical paper on the first-order method lower bound?}\Zhao{This is sidford paper has 80 citations, not sure if this is the correct ref} 
And these quantum algorithms cannot reduce the number of iterations. 
And this barrier also exists in quantum due to the $\Omega(\sqrt{m}/\epsilon)$ lower bound by \cite{ag18}.
%In \cite{ag18}, they proved a lower bound of , which means if we want to achieve sublinear in $m$, like most of the quantum first-order algorithms, then the $1/\epsilon$ dependence is inevitable.
For the quantum second-order algorithm \cite{kp20}, they achieved $\log(1/\epsilon)$ in their running time, which is indeed the advantage of the classical second-order method over the first-order method. However, the output of their algorithm are only guaranteed to be close to the feasible region, and the running time depends polynomially on the inverse of distance to feasible region. Hence, the error dependence of their algorithm is hard to specify\footnote{See page 9 in \cite{aggw17} for more discussions.}.
 
\paragraph{Linear-algebra computation: classical or quantum?}
In the classical interior-point methods, each iteration employs several linear-algebra computations. Some of them, like inverting a matrix or matrix-vector multiplication, can be sped-up via quantum linear-algebra techniques, e.g., the linear combination of unitaries (LCU) or the quantum singular value transformation (QSVT). However, there remains some operations that may not be computed in this way, e.g., flattening a matrix as a vector or stacking vectors as a matrix. Therefore, previous quantum algorithms \cite{kp20,antz21,kps21} uses the state tomography to transform quantum data to classical data, incurring a $1/\epsilon$ factor to achieve an $\epsilon$-error. Even worse, although quantum algorithms can solve the Newton system in poly-logarithmic time, the classical solution obtained by tomography is only an inexact solution of the linear system. This will significantly affect the convergence of the IPM iterations, and the previous quantum second-order methods \cite{kp20,antz21,kps21} only output \emph{approximately-feasible} SDP solution.

\paragraph{Nontrivial symmetrization.} 
For the feasible primal-dual central path method, the matrices $X$ and $S$ are required to be symmetric; otherwise, the convergence analysis cannot go through. Direct application of Newton method into central path leads to the following linear system:
 
\begin{align}\label{eq:newton_system}
    \begin{cases}
    \d S \cdot X + S\cdot \d X = \nu I - SX ~~~~ & \text{(complimentary slackness)}\\
    \langle \d X, A_i\rangle = 0~~~\forall i\in [m]~~~~ & \text{(primal feasibility)}\\
    \d S \in \mathrm{span}\{A_1,\dots, A_m\} ~~~~ & \text{(dual feasibility)}
    \end{cases}~~. 
\end{align}
However, this system in general yields nonsymmetric directions in the primal variable $\d X$ (see discussions in page 2 in \cite{mt00}). Hence, for the classical algorithms, in addition to solve the Newton linear system (Eq.~\eqref{eq:newton_system}), some nontrivial procedures should be applied to symmetrize the solution $\d X$. (See Section 2 in \cite{aho98} and Section 1 in \cite{mt00} and Section 4.5.1 in \cite{bn01} for more details.)\footnote{Notice that in Linear Programming $X$ and $S$ are both diagonal matrices, thus the solution of Eq.\eqref{eq:newton_system} is directly feasible. The standard primal-dual central path method can be applied (\cite{cls19}).}  However, \cite{kp20} does not adopt these procedures.\footnote{In their followup work \cite{kps21}, they use a correct symmetrization in their quantum second-order cone optimization algorithm.} %\Ruizhe{Baihe, can you %check if this sentence is correct? It could be \cite{bn01} have this step but \cite{kp20} only used a simplified version. Also, please find the page number of \cite{mt00}. Also, feel free to add more details on why non-symmetric will crash the analysis.}\Baihe{More precisely, \cite{kp20} does not adopt the classical algorithms (including the symmetrizing step) used in \cite{mt00}.}
%Therefore, \cite{kp20} did not handle this issue in their quantum second-order algorithm.}
Moreover, since the tomography algorithm they used can only output some $\ell_2$-approximations of the matrices $\d S$ and $\d X$, it will make both $X$ and $S$ nonsymmetric. %Therefore it is suspected that some extra efforts are needed to make their algorithm correct.

%\Zhao{please make sure mention the page number of \cite{aho98} and \cite{mt00}}
%\cite{aho98,mt00} described why primal-dual central path equation might generate non-symmetric solution $X,S$. But \cite{kp20} didn't handle this issue.

\section{Related Work}\label{sec:related_work}

%\paragraph{Classical SDP solvers}
%Table~\ref{table:runtime} summarizes the running time of the existing classical SDP solving algorithms and recent work \cite{hjstz22}. 

%\iffalse 
\begin{table*}[ht]
\small
\centering
 \begin{tabular}{ | l | p{5.0cm} | p{4cm} | l | l | l | l | }
 %{ | l | p{6.5cm} | l | l | l | l | l | }
    \hline
    \multirow{1}{*}{\bf Year} &
    \multirow{1}{*}{\bf Authors} &
    \multirow{1}{*}{\bf References} & \multirow{1}{*}{\bf Method} &  \multicolumn{1}{|c|}{\bf Runtime} \\ 
    \hline
    1979 & Shor, Yudin, Nemiroski, Khachiyan & \cite{s77,yn76,k80} & CPM & $n^8$ \\ \hline
    1988 & Khachiyan, Pavlovich, Erlikh, Nesterov, Nemirovski & \cite{kte88,nn89} & CPM &  $n^9$ \\ \hline
    1989 & Vaidya & \cite{v89} & CPM & $n^{6.746}$ \\ \hline
    1992 & Nesterov, Nemirovski & \cite{nn92} & IPM  & $n^{6.5} $ \\ \hline
    1994 & Nesterov, Nemirovski, Anstreicher & \cite{nn94, a00} & IPM  & $n^{10.75}$ \\ \hline
    2003 & Krishnan, Mitchell & \cite{km03} & CPM &  $n^{6.746}$ \\ \hline
    2015 & Lee, Sidford, Wong & \cite{lsw15} & CPM &  $n^6$ \\ \hline
    2020 & Jiang, Lee, Song, Wong & \cite{jlsw20} & CPM &  $n^6$ \\ \hline
    2020 & Jiang, Kathuria, Lee, Padmanabhan, Song & \cite{jklps20} & IPM &  $n^{5.246}$  \\ \hline
    2022 & Huang, Jiang, Song, Tao, Zhang & \cite{hjstz22} & IPM & $n^{4.746}$ \\ \hline 
  \end{tabular}
  \caption{\small Total running times of the algorithms for SDPs when $m=n^2$, where $n$ is the size of matrices, and $m$ is the number of constraints. Except for the last line, the rest of the table is the same as Table 1.2 in \cite{jklps20}. CPM denotes the cutting plane method, and IPM denotes the interior point method. The running times shown in the table hide $n^{o(1)}$, $m^{o(1)}$ and $\poly\log(1/\epsilon)$ factors, where $\epsilon$ is the accuracy parameter. In this table, we use the current best known upper bound of $\omega\approx 2.373$. This gives $n^{4.476} = n^{2\omega} = m^{\omega}$.} \label{table:runtime}
\end{table*}
%\fi 

%\Zhao{Cite some Ewin's paper}

%\paragraph{Quantum SDP solvers} 

Table~\ref{tab:qsdp} summarizes the existing quantum SDP solvers. The quantum first-order method algorithms \cite{bs17,aggw17,bkllsw19,ag18} is built on the Arora-Kale framework \cite{ak07} for SDP-solving via Multiplicative Weights Update (MWU) algorithm. The main observation is that the matrix
\begin{align*} 
    \rho:=\frac{ \exp(-\sum_{i=1}^m y_i A_i) }{\tr[ \exp( -\sum_{i=1}^m y_i A_i ) ]}
\end{align*} 
used in the multiplicative weights update is indeed a quantum Gibbs state, and hence a quantum Gibbs sampler can give speedup in $n$. The remaining part of the Arora-Kale framework is to find the violated constraints and update the dual solution $y$. Hence, they used some Grover-like techniques to achieve further quantum speedup in $m$. The quantum second-order method algorithm by Kerenidis and Prakash \cite{kp20} is based on the primal-dual central path framework, which requires to solve a Newton linear system in each iteration. \cite{kp20} used the quantum linear system solver of \cite{gslw19} to solver the Newton linear system in quantum and applied the tomography in \cite{kp17} to obtain the classical solution.%The fastest known quantum algorithm is by van Apeldoorn and Andr{\'a}s \cite{ag18} in the quantum operator input model that runs in time $O((\sqrt{m} + \sqrt{n}/\epsilon)\cdot \epsilon^{-4})$.
%which improved the previous works in two aspects: 1) proposing a more efficient Gibbs-sampler for estimating $\tr[A_i\rho]$; 2) using the Fast Quantum OR Lemma from \cite{bkllsw19} for the search problem. 
%), their algorithm runs in time $O((\sqrt{m} + \sqrt{n}/\epsilon)\cdot \epsilon^{-4})$, where $\epsilon$ is the approximation error. %\Zhao{The number is missing something?}

%and the quantum setting \cite{bs17,aggw17,bkllsw19,ag18,kp20}.
%\paragraph{Quantum linear algebra}

%There is a long line of work studying quantum linear algebra such as recommendation system \cite{kp17,t19,cgl+20}, principal component analysis \cite{t21}, linear regression \cite{gst22,cch+22}, and training neural networks 
%\cite{syz21}.

\begin{table}[ht]
\centering
\resizebox{0.9\textwidth}{!}{%
\begin{tabular}{|l|l|c|c|l|}
\hline
\textbf{\textbf{References}} &
  \multicolumn{1}{c|}{\textbf{Method}} &
  \textbf{\textbf{Input Model}} &
  \textbf{\textbf{Output}} &
  \multicolumn{1}{c|}{\textbf{Time Complexity}} \\ \hline \hline
\multirow{2}{*}{\cite{bs17}} &
  \multirow{2}{*}{1st. order} &
  \multirow{2}{*}{Sparse oracle} &
  \multirow{4}{*}{Feasible sol.} &
  \multirow{2}{*}{$\sqrt{mn}s^2\epsilon^{-32}$} \\
            &       &                                &           &                                                                    \\ \cline{1-3} \cline{5-5} 
\multirow{2}{*}{\cite{aggw17}} &
  \multirow{2}{*}{1st. order} &
  \multirow{2}{*}{Sparse oracle} &
   &
  \multirow{2}{*}{$\sqrt{mn}s^2 \epsilon^{-8}$} \\
            &       &                                &           &                                                                    \\ \hline
\multirow{3}{*}{\cite{bkllsw19}} &
  \multirow{3}{*}{1st. order} &
  Sparse oracle &
  \multirow{3}{*}{Decide feasibility} &
  $\sqrt{m}s^2 \epsilon^{-10}+\sqrt{n}s^2\epsilon^{-12}$ \\ \cline{3-3} \cline{5-5} 
            &       & \multirow{2}{*}{Quantum state} &           & \multirow{2}{*}{$\sqrt{m}\epsilon^{-\Omega(1)}$}       \\
            &       &                                &           &                                                                    \\ \hline
\multirow{5}{*}{\cite{ag18}} &
  \multirow{5}{*}{1st. order} &
  Sparse oracle &
  \multirow{5}{*}{Feasible sol.} &
  $\sqrt{m}s\epsilon^{-4}+\sqrt{n}s\epsilon^{-5}$ \\ \cline{3-3} \cline{5-5} 
            &       & \multirow{2}{*}{Quantum state} &           & \multirow{2}{*}{$\sqrt{m}\epsilon^{-4}+\epsilon^{-7.5}$} \\
            &       &                                &           &                                                                    \\ \cline{3-3} \cline{5-5} 
            &       & \multirow{2}{*}{QRAM}          &           & \multirow{2}{*}{$\sqrt{m}\epsilon^{-4}+\sqrt{n}\epsilon^{-5}$} \\
            &       &                                &           &                                                                    \\ \hline \hline 
\cite{kp20,kps21} & 2nd. order & QRAM                           & Infeasible sol. & $ n^{2.5} \mu\xi^{-2} \log (1/\epsilon)$  ($m=n^2$)\\ \hline \cite{antz21} & 2nd. order & QRAM & Feasible sol. & $n^{3.5}\epsilon^{-1}+n^4$ ($m=n^2$) \\
\hline
Ours (Thm.~\ref{thm:quantum_intro})        & 2nd. order & QRAM                           & Feasible sol. & $n^{3.5}\log(1/\epsilon)$ ($m=n^2$)                                         \\ \hline
\end{tabular}%

}
\caption{%\small 
Quantum algorithms for solving SDP. $\epsilon$ is the additive error of the output solution. $s$ is the row-sparsity of the input matrices. $\xi$ is the approximation feasibility. The sparse input oracle supports querying the $i$-th nonzero entry in the $j$-th row of an input matrix in superposition. The quantum state model splits each input matrix $A_i$ as a difference of PSD matrices (quantum states) and we can access to the density matrix of the state. The QRAM model assumes that the input matrices are stored in some data structures in QRAM that supports efficient quantum accessing to the matrix. Infeasible sol. means the algorithm only outputs a solution that can approximately satisfy the SDP constraints.}
\label{tab:qsdp}
\end{table}

\section{Our quantum technique}

To obtain a truly second-order quantum algorithm for solving SDP (Theorem~\ref{thm:quantum_intro}), we overcome the barriers stated in Section~\ref{sec:quantum_barrier} via a block encoding-based interior point method combined with our robust framework. In the following subsections, we will first describe our quantum algorithm and then demonstrate how it bypasses each barrier.

We define several notations related to number of iterations of our algorithm.
\begin{definition}\label{def:eta_T}
We choose $\eta$ as follows:
$
    \eta : = \frac{1}{n+2} .
$
We choose $T$ as follows
$
    T:= \frac{40}{\epsilon_N} \sqrt{n} \log (\frac{n}{\epsilon})
$.
\end{definition}

\begin{algorithm}[!ht]
\caption{Here, we briefly state an informal version of our algorithm. We put more details in later Algorithm~\ref{alg:qsdp_fast}. The input $\mathsf{A}$ has size $m \times n^2$. The input vector $b$ has length $m$. And the input matrix $C$ has size $n$ by $n$.
} 
\label{alg:quantum_ipm}
%\small
\begin{algorithmic}[1]
\Procedure{QSolveSDP}{$\mathsf{A}, b, C $}
\State \Comment{The input $\mathsf{A}, C, b$ are stored in QRAM}
\State Choose $\eta$ and $T$ according to Definition~\ref{def:eta_T}.
\State Find $y$ and store in QRAM\label{line:q_initialization} %\Comment{$y$ is an initial feasible dual}
\For {$t = 1 \to T$} %\Comment{Iterations of approximate barrier method}
	\State $\eta^{\new} \leftarrow \eta \cdot (1 + \frac{\epsilon_N}{20 \sqrt{n}})$
	\State Compute $\wt{S}^{-1}$ using QSVT and $\ell_2$-tomography \Comment{Step 1. of Sec.~\ref{sec:quantum_ipm_intro}}\label{line:compute_S}
	%\State Invert $\wt{S}$ classically and store $\wt{S}^{-1}$ in QRAM \label{line:compute_S_inv}
	\State Compute $\ket{g_\eta}$ and estimate the norm $\|g_\eta\|_2$ \Comment{Step 2. of Sec.~\ref{sec:quantum_ipm_intro}}\label{line:q_compute_gradient}
	%\hspace{8mm} $G \leftarrow$ \hspace{\fill} $\triangleright$ {$G \in \R^{m \times m}$} \label{line:q_compute_G}\\
	\State Compute $\wt{\delta}_y$ using quantum linear system solver and tomography \Comment{Step 3. of Sec.~\ref{sec:quantum_ipm_intro}}\label{line:compute_delta_y}
	\State $y^\new \leftarrow y + \wt{\delta}_y$\label{line:q_compute_y}
\EndFor
\EndProcedure
\end{algorithmic}
\end{algorithm}
 
\subsection{Block encoding-based interior point method}\label{sec:quantum_ipm_intro}
 
The framework of our quantum algorithm is the same as the classical algorithm in \cite{hjstz22}. We notice that most of the time-consuming steps involve matrix computations. Therefore, it is very natural to apply the recent block encoding and quantum linear algebra framework to speedup the interior point method.

\subsubsection{A brief overview of quantum linear algebra}
We provide the definition of the block encoding of a matrix.
\begin{definition}[Block encoding, informal]
Let $A$ be a matrix. We say a unitary matrix $U$ is a block encoding of $A$ if top-left block of $U_A$ is close to $A$ up to some scaling, i.e.,
\begin{align*}
    U_A\approx \begin{bmatrix}
    A/\alpha & \cdot\\
    \cdot & \cdot
    \end{bmatrix}
\end{align*}
\end{definition}
 
Via block encoding, we can transform a classical matrix to a quantum operator, which enables us to obtain exponential speedup in many linear algebra tasks by the quantum singular value transformation technique~\cite{gslw19}. %Before stating their results, we need to define some notations in quantum linear algebra. 
\iffalse
\begin{definition}[Vector state]
Let $v\in \R^n$ be a classical vector. Then, $v$ can be encoded by a $\lceil \log n\rceil$-qubit quantum state as follows:
\begin{align*}
    \ket{v}:=\frac{1}{\|v\|_2}\sum_{i=0}^{n-1}v_i\ket{i},
\end{align*}
where $\ket{i}$ is the $i$-th basis vector in the Hilbert space $\mathbb{C}^{2^{\lceil \log n \rceil}}$.
\end{definition}
\fi
%Let $\ket{i}$ be the $i$-th basis vector in a Hilbert space. Then, for a vector $x\in \R^n$, it can be encoded into a quantum state as 
%$
%    \ket{x}:=\frac{1}{\|x\|_2}\sum_{i=1}^{n}x_i\ket{i}
%$. 

Suppose we can efficiently implement the exact block encoding of a matrix $A$, and we can also efficiently prepare the vector state $\ket{x}$. Then, for matrix-vector multiplication, the state $\ket{Ax}=\frac{1}{\|Ax\|_2}\sum_{i=1}^{n} (Ax)_i \ket{i}$ can be $\epsilon$-approximated in time $\wt{O}(\log(1/\epsilon))$.\footnote{For simplicity, we hide the dependence of $\kappa(A)$ and $\mu(A)$, which is related to the singular values of $A$.} For linear system solving, the state $\ket{A^{-1}x}$ can also be $\epsilon$-approximated in time $\wt{O}(\log(1/\epsilon))$. In addition, their norms $\|Ax\|_2$ and $\|A^{-1}x\|_2$ could be approximated with $\epsilon$-multiplicative error. And it takes $\wt{O}(1/\epsilon)$ time.
 
\paragraph{Input model}
Our algorithm uses the QRAM (Quantum Random Access Memory) model, which is a common input model for many quantum algorithms (e.g., \cite{kp17, ag18,kp20}). More specifically, a QRAM stores the classical data and supports querying in superposition. The read and write operations of QRAM data of size $S$ will take $\wt{O}(S)$ time. Furthermore, \cite{kp17} showed that a QRAM can be extended a quantum data structure for storing matrices such that the exact block encoding can be prepared $\wt{O}(1)$ time. For the SDP inputs $\mathsf{A}, C$, we assume that they are already loaded into the QRAM such that their block encodings can be efficiently implemented. 

\subsubsection{Main steps of one iteration}\label{sec:quantum_alg_step_intro}
We will show how to use the block encoding and quantum linear algebra tools to implement \cite{hjstz22} general robust barrier method for SDP %(Algorithm~\ref{alg:robust_ipm:intro}) 
and achieve quantum speedup.  
 
\paragraph{Step 1: implement \textsc{ApproxSlackInverse}} 
The slack matrix $S$ is defined by 
\begin{align*} 
S(y) := \sum_{i=1}^m y_i A_i - C,
\end{align*} 
where the length-$m$ vector $y$ is the dual solution outputted by the previous iteration and we assume that it is stored in QRAM, $\mathrm{mat}(\cdot)$ means packing an $n^2$-length vector into an $n$-by-$n$ matrix. Since we assume quantum access to $A_i, C$ and $y$, by the linear combination of block-encodings, we can efficiently prepare the block-encoding of $S$. Then, for each $i\in [n]$, by the quantum linear system solver, we can prepare $\ket{S^{-1}e_i}=\ket{(S^{-1})_i}$ and also estimate the norm $\|(S^{-1})_i\|_2$ by $\ov{\|(S^{-1})_i\|_2}$. Then, we apply the $\ell_2$-tomography procedure in \cite{vcg22} to obtain a classical vector $v_i\in \R^{n}$ that is close to $\ket{(S^{-1})_i}$. Hence, by defining $(\wt{S}^{-1})_i:=\ov{\|(S^{-1})_i\|_2}\cdot v_i$, it holds that $\|(\wt{S}^{-1})_i - (S^{-1})_i\|_2\leq \epsilon_S\|(S^{-1})_i\|_2$. We repeat this procedure for all $i\in [n]$ and obtain a classical matrix $\wt{S}^{-1}\in\R^{n\times n}$ such that $\|\wt{S}^{-1}-S^{-1}\|_F\leq \epsilon_S\|S^{-1}\|_F$. To make $\wt{S}$ symmetric, we further let $\wt{S}^{-1}\gets \frac{1}{2}(\wt{S}^{-1}+(\wt{S}^{-1})^\top)$. The running time of this step is 
\begin{align*}
 {\cal T}_S=\wt{O}(n\cdot n\mu\kappa/\epsilon_S)\simeq \wt{O}(n^{2.5}\epsilon_S^{-1}),   
\end{align*}
which improves the classical running time $O(mn^2 + n^{\omega})$.
 
\paragraph{Step 2: implement \textsc{ApproxGradient}}
The gradient vector $g$ is defined as 
\begin{align*} 
g_\eta:=\eta \cdot b - \mathsf{A}\cdot \mathrm{vec}(S^{-1}).
\end{align*}
Note that the gradient is only used to compute $\delta_y:= -H^{-1}g_{\eta}$ in Step 3. Thus, it suffices to to prepare the state $\ket{g_{\eta}}$ and estimate its norm $\|g_{\eta}\|_2$. Using QRAM, we can efficiently implement the block-encoding for ${\sf A}'=\begin{bmatrix}{\sf A} &b\end{bmatrix}$. Since we compute $\wt{S}^{-1}$ in the classical form in Step 1, we can also efficiently prepare the state $\ket{s'}$ for $s'=\begin{bmatrix}-\mathrm{vec}(\wt{S}^{-1})\\\eta\end{bmatrix}$. By quantum linear algebra, we can prepare a state $\ket{\wt{g}_\eta(\wt{S})}$ that is $\epsilon_g$-close to $\ket{g_{\eta}(\wt{S})}$ in time 
\begin{align*}
    {\cal T}_{g,1}=\wt{O}(\mu({\sf A})\kappa({\sf A}))\simeq \wt{O}(\sqrt{m}),
\end{align*}
where $g_\eta(\wt{S})=\eta\cdot b - {\sf A}\cdot \mathrm{vec}(\wt{S}^{-1})$.
And we obtain an estimate $\ov{\|g_{\eta}(\wt{S})\|_2}$ of $\|g_\eta(\wt{S})\|_2$ within relative error $\epsilon_g$ in time 
\begin{align*}
    {\cal T}_{g,2}=\wt{O}(\mu({\sf A})\kappa({\sf A})\epsilon_g'^{-1})\simeq \wt{O}(\sqrt{m}\epsilon_g'^{-1}).
\end{align*}
Furthermore, if we define $\wt{g}_\eta(\wt{S}):=\ov{\|g_{\eta}(\wt{S})\|_2}\cdot \ket{\wt{g}_\eta(\wt{S})}$, then we can prove that
\begin{align*}
    \|\wt{g}_\eta(\wt{S}) - g_{\eta}\|_2 \lesssim O(\epsilon_g + \epsilon_g'+\epsilon_S)\|g_\eta\|_2,
\end{align*}
where we use the error guarantee of $\wt{S}^{-1}$.
%We may use quantum linear algebra and tomography to approximate each row of $\wt{S}^{-1}$. However, we find that this step is not robust enough and the approximation error should be at most $1/\sqrt{n}$, resulting in $\wt{O}(n^2)$-time for each row and $\wt{O}(n^3)$ in total, which is slower than the classical matrix inversion algorithm. Hence, we directly use the classical procedure to obtain the exact $\wt{S}^{-1}$ in $O(n^{\omega})$-time and store it in QRAM. Then, the remaining procedure is similar. By quantum matrix-vector multiplication and tomography, we will get a vector $\wt{g}_{\eta}$ such that $\|\wt{g}_\eta - g_{\eta}\|_2\leq \epsilon_2$ in time $\wt{O}(m/\epsilon_2^2)$. Thus, the total running time of this step is ${\cal T}_g:=\wt{O}(m/\epsilon_2^2 + n^{\omega})$.
 
\paragraph{Step 3: implement \textsc{ApproxDelta}}
In this step, we will compute $\delta_y:= -H^{-1}\cdot g_\eta$, where $H=(\mathsf{A}\cdot (S^{-1}\otimes S^{-1})\mathsf{A}^\top)$. Since ${\sf A}$ and $\wt{S}^{-1}$ are stored in QRAM, using the product and Kronecker product of block-encodings \cite{gslw19,cv20},   we can implement the block-encoding of $\wt{H}=(\mathsf{A}\cdot (\wt{S}^{-1}\otimes \wt{S}^{-1})\mathsf{A}^\top)$. In order to reduce the block-encoding factor of $\wt{H}$, we first implement the block-encoding of $\wt{S}^{-1/2}\otimes \wt{S}^{-1/2}$, by classically compute $\wt{S}^{-1/2}=(\wt{S}^{-1})^{1/2}$. Then, we implement the block-encoding of $W={\sf A}\cdot (\wt{S}^{-1/2}\otimes \wt{S}^{-1/2})$. Since $\wt{H}=WW^\top$, the block-encoding of $\wt{H}$ can be efficiently implemented.  In Step 2, $\ket{\wt{g}_{\eta}(\wt{S})}$ are prepared. Then, we can apply the quantum linear system solver to obtain the state close to $\ket{\wt{H}^{-1}\wt{g}_{\eta}(\wt{S})}$ and also estimate the norm $\|\wt{H}^{-1}\wt{g}_{\eta}(\wt{S})\|_2$. Then, we apply the tomography to obtain the classical form of $\wt{H}^{-1}\wt{g}_{\eta}(\wt{S})$. Let $\wt{\delta}_y$ denote the classical vector we compute in this step. We can show that 
\begin{align*}
    \|\wt{\delta}_y-(-\wt{H}^{-1}\wt{g}_{\eta}(\wt{S}))\|_2\leq \epsilon_\delta \|\wt{H}^{-1}\wt{g}_{\eta}(\wt{S}))\|_2
\end{align*}
in time 
\begin{align*}
    {\cal T}_{\delta}= \wt{O}(m\mu({\sf A})\epsilon_\delta^{-1}+m\mu(S)\epsilon_\delta^{-1})\simeq \wt{O}(mn\epsilon_\delta^{-1}),
\end{align*}
where we use the fact that $\mu({\sf A})\leq \sqrt{m+n^2}=O(n)$.
%Together with $\ket{\wt{g}_\eta}$, we can compute  $\wt{\delta}_{y}$ such that $\|\wt{\delta}_y- (-\wt{H}^{-1}\wt{g}_\eta)\|_2\leq \epsilon_3$ in time ${\cal T}_\delta := \wt{O}(m/\epsilon_3^2)$.
 
\paragraph{Running time per iteration}
Putting three steps together, we get that the total running time per iteration of our algorithm is 
\begin{align*} 
{\cal T}_{\rm iter} = {\cal T}_S + {\cal T}_{g,2}+{\cal T}_\delta\simeq \wt{O}(n^{2.5}\epsilon_S^{-1} + \sqrt{m}\epsilon_g^{-1}+mn\epsilon_\delta^{-1}).
\end{align*}
 
\subsection{Overcoming the quantum barriers}
In this part, we discuss how our quantum algorithm bypasses each barrier in Section~\ref{sec:quantum_barrier}.
 
\paragraph{Error dependence barrier } We first note that by the analysis of our robust interior-point framework (Theorem~\ref{thm:approx_central_path_dual}), the number of iterations will be $\wt{O}(\sqrt{n}\log(1/\epsilon))$ as long as the Newton step size satisfies $g_\eta^\top H^{-1} g_\eta \leq \epsilon_N^2$ for some constant $\epsilon_N$ in each iteration. As shown in Section~\ref{sec:quantum_alg_step_intro}, the running time per iteration is depends linearly on $\epsilon_S^{-1},\epsilon_g^{-1},\epsilon_\delta^{-1}$. Fortunately, by the recently proposed robust framework (Lemma~\ref{lem:invariant_newton}), constant accuracies\footnote{They actually depend on the condition number of the matrices. See Section~\ref{sec:q_combine} for details.} are enough:
 
\begin{lemma}[Quantum time cost per iteration, informal]
In each iteration, we can take some properly chosen $\epsilon_S,\epsilon_g,\epsilon_\delta$ such that \textbf{Condition 0-4} in the robust IPM framework (Lemma~\ref{lem:invariant_newton}) are satisfied and each iteration runs in  
$
    \wt{O}(m^{1.5}+n^{2.5})
$ time
\end{lemma}
 
Therefore, the total running time of our quantum SDP solving algorithm is $$\wt{O}(\sqrt{n}(mn+n^{2.5})\log(1/\epsilon)).$$ It depends only logarithmically on the approximation error, which is a truly second-order algorithm.
 
\paragraph{Quantum/classical data transformation barrier } We can output the classical solution by running tomography in intermediate steps. For time-consuming steps, we use the QSVT to speed up. And for the steps that are hard to be implemented in the QSVT (e.g., vectorizing $\wt{S}^{-1}$), we directly compute them classically with low-accuracy tomography procedure. In this way, the robust IPM framework can tolerate the ``quantum-inherent'' error, and our quantum algorithm can still output \emph{feasible, high-accuracy} SDP solution.  %Our algorithm can naturally bypass it because in each step of the iteration, we will use the tomography procedure to extract the classical information from the quantum objects we compute. And we can show that doing this will not make the global error dependence worse. 
 
\paragraph{Symmetrization barrier } 
We bypass the barrier by solving SDP in the dual space instead of the primal space, where we only care about the symmetrization of $\wt{S}$, the approximation of the slack matrix $S$. It is easy to make $\wt{S}$ symmetric by averaging $\wt{S}$ and $\wt{S}^\top$. This is because the definition of $S$ guarantees that the true $S$ is symmetric. Hence, the classical interior point method does not have this problem and our symmetrization will only make the approximation smaller. However, we note that this symmetrization trick cannot directly apply to \cite{kp20}'s algorithm. Because they solved SDP also in the primal space and the true solution $\d X$ of the newton linear system may not be symmetric, which means averaging $\wt{\d X}$ and $(\wt{\d X})^\top$ could result in some error that is hard to control. A possible way may be to quantize the classical symmetrization procedures given in \cite{aho98}.
%We solve SDP in the dual space instead of the primal space, where we only care about the symmetrization of $\wt{S}$. It can be fixed by averaging $\wt{S}$ and $\wt{S}^\top$, since by definition the true $S$ is symmetric. 
%Hence, our symmetrization will only make the error smaller. 
%However, this trick cannot work for \cite{kp20}'s algorithm. Because they solved SDP in the primal space and the true solution $\d X$ of the newton linear system may not be symmetric, which means averaging $\wt{\d X}$ and $(\wt{\d X})^\top$ could result in some error that we cannot control. A possible solution is to quantize the classical symmetrization procedures \cite{aho98}.

%\addcontentsline{toc}{section}{References}
\ifdefined\isarxivversion
\else
\fi
%\bibliography{ref}

\iffalse
\newpage
{\hypersetup{linkcolor=black}
\tableofcontents
}
\fi

%\newpage
\section{Preliminary}\label{sec:preli}

%\subsection{Notations}

We define several basic notations here. 
We say $\tr[\cdot]$ is the trace of a matrix.

We use $\| \cdot \|_2$ to denote spectral/operator norm of a matrix. Let  $\|\cdot\|_F$ be the Frobenious norm of a matrix. We use $\| \cdot \|_1$ to represent Schatten $1$-norm of matrix.

For a symmetric matrix $X \in \R^{n \times n}$, we say it is positive semi-definite (PSD, denoted as $X \succeq 0$) if for any vector $u \in \R^n$, $u^{\top} X u \geq 0$. Similarly, we define positive definite via $\succ$ and $>0$ notations.

%For a symmetric matrix $A \in \R^{n \times n}$, we say it is positive definite (PD, denoted as $A \succ 0$) if for any vector $x \in \R^n$, $x^{\top} A x > 0$. 

%We use $\mathbb{S}^{n \times n}_{>0}$ to denote the set of all $n$-by-$n$ symmetric positive definite matrices and $\mathbb{S}^{n \times n}_{\geq 0}$ to denote the set of all $n$-by-$n$ symmetric semi-definite matrices.

We say $\lambda(B)$ are the eigenvalues of $B$.

We use $x_{[i]}$ to denote the $i$-th largest entry of vector $x$.

\iffalse 
For a matrix $A \in \R^{m \times n}$, and subsets $S_1\subseteq [m], S_2\subseteq [n]$, we define $A_{S_1,S_2} \in \R^{|S_1|\times |S_2|}$ to be the submatrix of $A$ that only has rows in $S_1$ and columns in $S_2$. We also define $A_{S_1,:} \in \R^{|S_1| \times n}$ to be the submatrix of $A$ that only has rows in $S_1$, and $A_{:,S_2} \in \R^{m \times |S_2|}$ to be the submatrix of $A$ that only has columns in $S_2$.
\fi

%\paragraph{Matrix related operations.}
%For two matrices $A,B \in \R^{m \times n}$, we define the matrix inner product $\langle A, B \rangle := \tr[A^{\top} B]$.

We use $\vect[]$ to denote matrix vectorization.
\iffalse 
: for a matrix $A \in \R^{m \times n}$, $\vect[A] \in \R^{mn}$ is defined to be $\vect[A]_{(j-1)\cdot n + i} = A_{i,j}$ for any $i \in [m]$ and $j \in [n]$, i.e., \begin{align*}
   \vect[A]  = \begin{bmatrix} A_{:,1} \\ \vdots\\ A_{:,n}\end{bmatrix} \in \R^{mn}.
\end{align*}
\fi

Let $\otimes$  denote the Kronecker product. %: for matrices $A \in \R^{m \times n}$ and $B \in \R^{p \times q}$, $A \otimes B \in \R^{pm \times qn}$ is defined to be $(A\otimes B)_{p(i-1) + s, q(j-1) + t} = A_{i,j} \cdot B_{s,t}$ for any $i \in [m]$, $j \in [n]$, $s \in [p]$, $[t] \in [q]$. %, i.e., 
\iffalse 
\begin{align*}
     A \otimes B  = 
     \begin{bmatrix} 
     A_{1,1} \cdot B & A_{1,2} \cdot B & \dots & A_{1,n} \cdot B\\ 
     A_{2,1} \cdot B & A_{2,2} \cdot B & \cdots & A_{2,n} \cdot B \\
     \vdots & \vdots & \ddots &\vdots \\
     A_{m,1} \cdot B & A_{m,2} \cdot B & \dots & A_{m,n} \cdot B
     \end{bmatrix} 
     \in \R^{pm \times qn}.
 \end{align*}

\fi

%\begin{definition}[Stacking matrices]%\label{def:stacking_matrices}
%Let $A_1, A_2, \cdots, A_m \in \R^{n \times n}$ be $m$ symmetric matrices.

The $\mathsf{A} \in \R^{m \times n^2}$ is a matrix where $i$-th row is $\vect[A_i]$, for all $i \in [m]$. % denote the matrix that is constructed by stacking the $m$ vectorizations $\vect[A_1], \vect[A_2], \cdots \vect[A_m] \in \R^{n^2}$.
%\end{definition}

Let $\Tmat(x,y,z)$ denote the time of multiplying an $x \times y$ matrix with another $y \times z$ matrix.

For $r \in \R_+$, let $\omega(r) \in \R_+$ denote the value that for all positive integer $m$, $\Tmat(m,m,m^r) = O(m^{\omega(r)})$.

Eq.~\eqref{eq:sdp_primal} is the primal form of SDP. Here we define dual form:
\begin{definition}[The dual forumlation of SDP]\label{def:sdp_dual}
Suppose we are given a symmetric matrices $C \in \R^{n \times n}$. There are also $m$ constraints matrices $A_1,\dots,A_m $, and each of them has size $n$ by $n$. Suppose we are also given a length-$m$ vector $b $. Our goal is to solve this problem:
\begin{align}\label{eq:sdp_dual}
    \min_{y \in \R^m} & ~ \langle b , y \rangle ~~~ \notag\\ 
    \mathrm{~s.t.} & ~ S = \sum_{i=1}^m y_i A_i - C, \\
    & ~~~~ S \succeq 0.  \notag
\end{align}
\end{definition}

We state two tools from previous work. We remark those two tools are only used in quantum SDP solver. In the classical SDP solver, they don't need to use these tools \cite{hjstz22}.
\begin{theorem}[\cite{w73,mz10}]\label{thm:wedin_spectral}
Let $A \in \R^{m \times n}$ and $B= A + E$. Then
\begin{align*}
    \| B^\dagger - A^\dagger \| \leq \sqrt{2} \max\{ \| A^\dagger \|^2, \| B^\dagger \|^2 \} \cdot \| E \|.
\end{align*}
\end{theorem}

\begin{theorem}[\cite{w73,mz10}]\label{thm:wedin_frobenius}
Let $A \in \R^{m \times n}$ and $B= A + E$. Then
\begin{align*}
    \| B^\dagger - A^\dagger \|_F \leq \max\{ \| A^\dagger \|^2, \| B^\dagger \|^2 \} \cdot \| E \|_F.
\end{align*}
If $\rank(A) = \rank(B)$, we have
\begin{align*}
    \| B^\dagger - A^\dagger \|_F \leq \| A^\dagger \|\cdot \| B^\dagger \| \cdot \| E \|_F
\end{align*}
\end{theorem}

%\input{combine} %This is a classical section, we seems don't need it at all.
%\input{correct}
%\input{time_per_iter}
%\input{amortize}

%\section{Our Robust SDP-based IPM Error Analaysis}
\label{sec:correct_previous}

\section{Revisit of Robust Newton Step}\label{sec:our_roubst_ns}
In this work, one of our contributions is that we provide a more robust version of the potential function that considers the perturbation of gradients. 

Classical interior point literature \cite{r01} controls the potential function when the exact Newton step is taken. \cite{jklps20} introduces small error in the the framework. Recently, \cite{hjstz22} provides the most general framework,

\iffalse 
\begin{lemma}[\cite{r01} exact framework]
Given $ \epsilon_N \in (0, 1/10)$ and $\eta > 0$. Suppose that there is 
\begin{itemize}
    \item {\bf Condition 0.} a feasible solution $y \in \R^m$ satisfies  $ \| g(y,\eta) \|_{ H(y)^{-1} } \leq \epsilon_N$,
\end{itemize}
Then $\eta^{\new} = \eta (1 + \frac{\epsilon_N}{20 \sqrt{n}} ) $ and $y^{\new} = y - H(y)^{-1} g(y,{\eta^{\new} })$ satisfy \begin{align*}
\| g (y^{\new}, \eta^{\new} ) \|_{  H(y^{\new})^{-1} } \leq \epsilon_N.
\end{align*}
\end{lemma}

Further, \cite{jklps20} relax the exact $H$ to a PSD approximation version,
\begin{lemma}[\cite{jklps20} semi-robust framework]
Given parameters $ \epsilon_N \in (0, 1/10)$, $\eta > 0$ and  $ \alpha_H \in [1,1+10^{-4}]$. Suppose that there is 
\begin{itemize}
    \item {\bf Condition 0.} a feasible solution $y \in \R^m$ satisfies  $ \| g(y,\eta) \|_{ H(y)^{-1} } \leq \epsilon_N$,
    \item {\bf Condition 1.} a positive definite matrix $\wt{H} \in \mathbb{S}^{n \times n}_{>0}$ satisfies $\alpha_H^{-1} H(y) \preceq \wt{H} \preceq \alpha_H  H(y)$.
\end{itemize}
Then $\eta^{\new} = \eta (1 + \frac{\epsilon_N}{20 \sqrt{n}} ) $ and $y^{\new} = y - \wt{H}^{-1} g(y,{\eta^{\new} })$ satisfy \begin{align*} 
\| g (y^{\new}, \eta^{\new} ) \|_{  H(y^{\new})^{-1} } \leq \epsilon_N.
\end{align*}
\end{lemma}

We propose a more general framework in the following Lemma and we believe it will be useful in the future optimization task for semi-definite programming.
\fi 

\begin{lemma}[\cite{hjstz22} robust Newton step]\label{lem:invariant_newton}
Let $c_0 = 10^{-4}$. Given any parameters 
\begin{itemize}
    \item $ \alpha_S \in [1, 1+ c_0 ]$,
    \item $ \alpha_H \in [1, 1+ c_0 ]$,
    \item $ \epsilon_g   \in [0, c_0 ]$,
    \item $  \epsilon_{\delta} \in [0, c_0]$,
    \item $ \epsilon_N  \in [0,c_0 ]$,
    \item $\eta >0$.
\end{itemize}

We assume 
\begin{itemize}
    \item {\bf Condition 0.} A length-$m$ vector $y \in \R^m$ (feasible dual solution) satisfies 
    \begin{align*} 
        \| g(y,\eta) \|_{ H(y)^{-1} } \leq \epsilon_N.
    \end{align*}
    \item {\bf Condition 1.} A $n \times n$ symmetric (positive definite) matrix $\wt{S}$ satisfies 
    \begin{align*}
        \alpha_S^{-1} \cdot S(y) \preceq \wt{S} \preceq \alpha_S \cdot S(y).
    \end{align*}
    \item {\bf Condition 2.} A $n \times n$ symmetric (positive definite) matrix $\wt{H}$ satisfies 
    \begin{align*} 
        \alpha_H^{-1} \cdot H( \wt{S} )  \preceq \wt{H} \preceq \alpha_H \cdot H( \wt{S} ) .
    \end{align*}
    \item {\bf Condition 3.} A length-$m$ vector $\wt{g}$ satisfies \begin{align*} 
        \|\wt{g} - g(y,\eta^{\new} )\|_{H(y)^{-1}} \leq \epsilon_g \cdot \|g(y,\eta^{\new})\|_{H(y)^{-1}} .
    \end{align*}
    \item {\bf Condition 4.} A length-$m$ vector $\wt{\delta}_y$ satisfies 
    \begin{align*} 
        \|\wt{\delta}_y - (- \wt{H}^{-1} \wt{g})\|_{H(y)} \leq \epsilon_{\delta} \cdot \| \wt{H}^{-1} \wt{g} \|_{H(y)} .
    \end{align*}
\end{itemize}
Then following the update rule of $\eta$ and $y$: $\eta^{\new} = \eta \cdot (1 + \frac{\epsilon_N}{20 \sqrt{n}} ) $ and $y^{\new} = y + \wt{\delta}_y$, we have $y^{\new}$ and $\eta^{\new}$ satisfy the following:
\begin{align*}
\| g (y^{\new}, \eta^{\new} ) \|_{  H(y^{\new})^{-1} } \leq \epsilon_N.
\end{align*}
\end{lemma}

\begin{remark} \label{rmk:l_2_cond}
{\bf Condition 3 and 4} can be derived by the following $\ell_2$ guarantees:
\begin{itemize}
\item $\|\wt{g} - g(y , \eta^{\new} )\|_2 \leq \epsilon_g \cdot \|g (y, \eta^{\new})\|_2 / \kappa(H(y))$,
\item $\|\wt{\delta}_y - (- \wt{H}^{-1} \wt{g})\|_2 \leq \epsilon_{\delta} \cdot \| \wt{H}^{-1} \wt{g} \|_{2} /\kappa(H(y))$. %\Zhao{This $\lambda_{\max}$ might be $\kappa$.}
\end{itemize}
\end{remark}

\begin{theorem}[Robust barrier method, Theorem 11.16 in \cite{hjstz22}]\label{thm:approx_central_path_dual}
Consider a semidefinite program in Eq~\eqref{eq:sdp_dual}. 
%Assume that any feasible solution $X \in \mathbb{S}^{n \times n}_{\geq 0}$ satisfies $\| X \|_{2} \leq R$. 
Suppose in each iteration, the $\wt{S}, \wt{H},\wt{g},\wt{\delta}$ computed in Line~\ref{lin:approx_S}, Line~\ref{lin:approx_H}, Line~\ref{lin:approx_g}, Line~\ref{lin:approx_delta} of Algorithm~\ref{alg:robust_ipm} satisfies Condition 1, 2, 3, 4 in Lemma~\ref{lem:invariant_newton}.

Let $y$ denote a length $m$ vector (can be viewed as feasible initial solution). Suppose that $y$ satisfies the invariant $g(y, \eta)^\top H(y)^{-1} g(y,\eta) \leq \epsilon_N^2$, for any error parameter $0 < \epsilon \leq 0.01$ and Newton step size $\epsilon_N$ satisfying  $\sqrt{\epsilon} < \epsilon_N \leq 0.1$, Algorithm~\ref{alg:robust_ipm} outputs,  in $T = 40 \epsilon_N^{-1} \sqrt{n} \log(n/\epsilon)$ iterations,  a vector $y \in \mathbb{R}^{m}$ that satisfies 
\begin{align}\label{eq:dual_optimality_gap}
\langle b , y  \rangle \leq \langle b , y^* \rangle + \epsilon^2 .  
\end{align}
where $y^*$ is an optimal solution to the dual formulation \eqref{eq:sdp_dual}. 

Further, consider the Algorithm~\ref{alg:robust_ipm}, we have 
\begin{align}\label{eq:promise_sdp}
\| S^{-1/2} S^{\new} S^{-1/2} - I \|_F \leq 1.1 \cdot \epsilon_N
\end{align}
holds for each iteration.
\end{theorem}

\begin{algorithm}[!ht]
\caption{We restate the framework of \cite{hjstz22} method. Let $C \in \R^{n \times n}$. Let $A_1, A_2, \cdots, A_m$ denote a list of $m$ matrices where each matrix has size $n \times n$. Let $b$ denote a length-$m$ vector. Let $\mathsf{A}$ denote a matrix that has size $m \times n^2$.}
\label{alg:robust_ipm}
\begin{algorithmic}[1]
\Procedure{\textsc{SDPFrameWork}}{$\{ A_i \}_{i=1}^m$, $b, C, m, n$} 
%\State \Comment{Initialization} 
\State Choose $\eta$ and $T$ according to Definition~\ref{def:eta_T} %$\eta \leftarrow \frac{1}{n+2}$, ~~$T \leftarrow \frac{40}{\epsilon_N} \sqrt{n} \log (\frac{n}{\epsilon})$
\State Find  $y \in \R^m$ according to \cite{hjstz22} % Lemma~\ref{lem:initialization} initial feasible dual
\Comment{Condition 0 in Lemma~\ref{lem:invariant_newton}}
\For {$t = 1 \to T$} %{\bf do}  %\Comment {Iterations of approximate barrier method}
	\State $\eta^{\new} \leftarrow \eta \cdot (1 + \frac{\epsilon_N}{20 \sqrt{n}})$
	\State $\wt{S} \leftarrow$ ~Approximate Slack \Comment{Condition 1 in Lemma~\ref{lem:invariant_newton}}\label{lin:approx_S}
	\State $\wt{H} \leftarrow$ Approximate Hessian \Comment{Condition 2 in Lemma~\ref{lem:invariant_newton}}\label{lin:approx_H}
	\State $\wt{g} \leftarrow$ ~Approximate Gradient \Comment{Condition 3 in Lemma~\ref{lem:invariant_newton}}\label{lin:approx_g}
	\State $\wt{\delta}_y \leftarrow$  Approximate Delta \Comment{Condition 4 in Lemma~\ref{lem:invariant_newton}}\label{lin:approx_delta} 
	\State $y^{\new} \leftarrow y + \delta_y$ 
	\State $y \leftarrow y^{\new}$ 
	%\Comment{Update variables}
\EndFor
\State Return a solution via \cite{hjstz22}
\EndProcedure
\end{algorithmic}
\end{algorithm}

\section{Quantum preliminaries}
\subsection{Basics of quantum computing}
In this section, we give a brief overview of quantum computing. We refer to the standard textbook \cite{nc11} for a more comprehensive introduction of quantum information and quantum computing.

\paragraph{Quantum states} Let $\mathcal{H}$ be a finite-dimensional Hilbert space. Quantum states over $\mathcal{H}$ are positive semi-definite operator from $\mathcal{H}$ to $\mathcal{H}$ with unit trace. When $\mathcal{H}=\mathbb{C}^2$, a quantum state is called a qubit, which can be represented by a linear combination of two standard basis vectors $\ket{0}$ and $\ket{1}$. More generally, for an $n$-dimensional Hilbert space $\mathcal{H}$, a pure state can be represented as a unit column vector  $\ket{\phi}\in \mathbb{C}^n$ and the standard basis vectors are written as $\{\ket{i}\}$ for $0\leq i\leq n-1$. If a state $\ket{\phi}$ is a linear combination of several standard basis states $\ket{i}$, then we say $\ket{\phi}$ is in ``superposition''. For $\ket{\phi}$ We use $\bra{\phi}$ to denote the complex conjugate of $\ket{\phi}$, i.e., $\bra{\phi} = (c_0^*,\dots,c_{n-1}^*)$, where $c_i$ is the coefficient of $\ket{\phi}$ in the $i$-th standard basis.

\paragraph{Quantum unitaries} Quantum operations on quantum states are defined as unitary transformations.
\begin{definition}[Unitary transformation]
A unitary transformation $U$ for an $n$-qubit quantum state is an isomorphism in the $2^n$-dimensional Hilbert space, which can be represented as a $2^n$-by-$2^n$ complex matrix with $UU^\dagger = U^\dagger U = I$ where $U^\dagger$ is the Hermitian adjoint of $U$. 
\end{definition}
We use $U\ket{\phi}$ to denote applying a quantum unitary $U$ on a quantum state $\ket{\phi}$. We also use $U\otimes V$ to denote the \emph{tensor product} of $U,V$ if they are quantum unitaries. 

\paragraph{Measurement} For any quantum algorithm, the input is a quantum state, the algorithm procedure can be described by a quantum unitary, and the output is still a quantum state. To extract useful information from the output state, we can \emph{measure} it. Mathematically, measuring the quantum state is equivalent to a sampling process. Suppose $\ket{\phi}=\sum_{i=0}^{n-1}c_i \ket{i}$. Then when we measure it in the standard basis, with probability $|c_i|^2$ we will get the outcome $i$.

\subsection{Tools: Quantum linear algebra}
\subsubsection{Storing data in QRAM}
Quantum random access memory (QRAM) is a commonly-used model in quantum computing that assumes the quantum computer can access classical data in superposition.
\begin{definition}[Quantum random access memory (QRAM), \cite{gio08}]
A quantum random access memory is a device that stores indexed data $(i,x_i)$ for $i\in \N$ and $x_i\in \R$ (with some bit precision). It allows to query in the form $\ket{i}\ket{0}\mapsto \ket{i}\ket{x_i}$. Each read/write/update operation has $\wt{O}(1)$ cost.
\end{definition}

Furthermore, there are some ways to store classical vectors and matrices in QRAM using some quantum data structures developed in \cite{kp17}, such that for any vector or any row of a matrix, the vector state
\begin{align*}
    \ket{v}:=\frac{1}{\|v\|_2}\sum_{i=0}^{n-1}v_i\ket{i}
\end{align*}
can be prepared in $\mathrm{polylog}(n)$ time. In addition, for a matrix $A$, its row norm state $\sum_{i=0}^n \|A_i\|_2\ket{i}$ can also be prepared in $\mathrm{polylog}(n)$ time. 

We now introduce a commonly used matrix parameter in quantum linear algebra:
\begin{definition}[Matrix parameter for QRAM]
For a matrix $A\in \R^{n\times m}$, for $p\in [0,1]$, let $s_p(A):=\max_{i\in [n]}\sum_{j\in [m]}A_{i,j}^p$. Then, we define a parameter $\mu(A)$ as follows:
\begin{align*}
    \mu(A):=\min_{p\in [0,1]} \left\{\|A\|_F, \sqrt{s_{2p}(A),s_{1-2p}(A^\top)}\right\}.
\end{align*}
\end{definition}

\subsubsection{Retrieving data from quantum to classical}
In order to transform a quantum vector state $\ket{v}$ to a classical vector, we need to apply the state tomography procedure \cite{kp20,klp19,vcg22}. We state a version with an $\ell_2$-approximation guarantee.
\begin{theorem}[Vector state tomography with $\ell_2$ guarantee, {\cite[Theorem 23]{vcg22}}]\label{thm:tomography}
Given access to a unitary $U$ such that $U\ket{0}=\ket{x}=\sum_{i=0}^{n-1} x_i\ket{i}$ for some $x\in \R^n$, there is a tomography algorithm that outputs a vector $\wt{x}\in \R^d$ with probability $1-1/\poly(d)$ using $\wt{O}(d/\epsilon)$ conditional applications of $U$ and its inverse and $\wt{O}(d/\epsilon)$ additional quantum gates such that $\|x-\wt{x}\|_2\leq \epsilon$. 
\end{theorem}

\subsubsection{Quantum linear algebra with block encodings}
Block encoding is a way to efficiently transform a classical matrix to a quantum operator that enables quantum computers to speedup several linear algebra computations. 

The definition of block encoding is as follows:
\begin{definition}[Block encoding]
Let $A\in \C^{2^w\times 2^w}$ be a matrix. We say a unitary matrix $U\C^{2^{(w+a)}\times 2^{(w+a)}}$ is a $(\alpha, a, \epsilon)$block encoding of $A$ if 
\begin{align*}
    \left\| A - \alpha (\langle 0|^a \otimes I )U( |0\rangle^a \otimes I)  \right\| \leq \epsilon,
\end{align*}
i.e., 
$$
    U_A\approx \begin{bmatrix}
    A/\alpha & \cdot\\
    \cdot & \cdot
    \end{bmatrix}.
$$
\end{definition}

Kerenidis and Prakash \cite{kp17} showed that for a matrix stored in a quantum data structure, its block encoding can be efficiently implemented.

\begin{theorem}[Block encoding construction from QRAM, \cite{kp17}]
Given quantum access to a matrix $A\in \R^{n\times n}$, a $(\mu(A), O(\log(n)), 0)$-block encoding of $A$ can be implemented in $\wt{O}(1)$ time.  
\end{theorem}

We state a tool from previous work \cite{gslw19}.
\begin{theorem}[Product of block encoded matrices, Lemma 53 in \cite{gslw19}]\label{thm:be_product}
We use $U_A$ to represent an $(\alpha, a, \epsilon_A)$-block encoding of a matrix $A$ which can be constructed in time $T_A$, and $U_B$ be a $(\beta, b, \epsilon_B)$-block encoding of a matrix that can be constructed in time $T_B$. 

Then, we can implement an $(\alpha \beta, a+b, \alpha \epsilon_B + \beta\epsilon_A)$-block encoding of $AB$ in time $O(T_A+T_B)$.
\end{theorem}

\begin{lemma}[Product of preamplified block-matrices \cite{lc19}]\label{lem:pre_amp_prod}  
Let $A \in \R^{m\times n}$ and $B \in \R^{n\times k}$ such that $\|A\|\leq 1, \|B\|\leq 1$. If $\alpha \geq 1$ and $U$ is an $(\alpha,a,\delta)$-block-encoding of $A$ that can be implemented in time $T_U$; $\beta \geq 1$ and $V$ is a $(\beta,b,\epsilon)$-block-encoding of $B$ that can be implemented in time $T_V$, then there is a $(2,a+b+2,\sqrt{2}(\delta +\epsilon+\gamma))$-block-encoding of $AB$ that can be implemented in time $O((\alpha(T_U +a)+\beta(T_V +b))\log(1/\gamma))$. 
\end{lemma}

We state a tool from previous work \cite{cv20}.
\begin{theorem}[Kronecker product of block encoded matrices, Lemma 1 in \cite{cv20}]\label{thm:be_construct_kronecker}
We use $U_A$ to represent an $(\alpha, a, \epsilon_A)$-block encoding of a matrix $A$ which can be constructed in time $T_A$, and $U_B$ be a $(\beta, b, \epsilon_B)$-block encoding of a matrix that can be constructed in time $T_B$. 

Then, we can implement an $(\alpha \beta, a+b, \alpha \epsilon_B + \beta\epsilon_A+\epsilon_A\epsilon_B)$-block encoding of $A\otimes B$ in time $O(T_A+T_B+\log(n))$.
\end{theorem}

\begin{theorem}[Quantum linear system solver, \cite{cgj18,gslw19}] \label{thm:qlsa} 
Let $A \in \R^{n\times n}$ be a matrix with non-zero eigenvalues in the interval $[-1, -1/\kappa] \cup [1/\kappa, 1]$. Given an implementation of an $(\mu, O(\log n), \delta )$ block encoding for $A$ in time $T_{U}$ and a procedure for preparing state $\ket{b}$ in time $T_{b}$, 
\begin{enumerate} 
\item If $\delta \leq \frac{\epsilon}{\kappa^{2} \poly\log(\kappa/\epsilon)}$ then a state $\epsilon$-close to $\ket{A^{-1} b}$ can be generated in time 
\begin{align*}
    (T_{U} \kappa \mu+ T_{b} \kappa) \cdot \poly\log(\kappa \mu /\epsilon).
\end{align*}
 
\item If $\delta \leq \frac{\epsilon}{2\kappa}$ then a state $\epsilon$-close to $\ket{A b}$ can be generated in time 
\begin{align*}
    (T_{U} \kappa \mu+ T_{b} \kappa) \cdot \poly\log(\kappa \mu /\epsilon).
\end{align*}

\item For $\epsilon>0$ and $\delta$ as in parts 1 and 2 and $\mathcal{A} \in \{ A, A^{-1} \}$, an estimate $\Lambda$ such that $\Lambda \in (1\pm \epsilon) \| \mathcal{A} b\|$ with probability $(1-\delta)$ can be generated in time
\begin{align*}
    (T_{U}+T_{b}) (\kappa \mu/ \epsilon) \cdot \poly\log(\kappa \mu/\epsilon). 
\end{align*}
\end{enumerate} 
\end{theorem} 

\iffalse
\begin{theorem}[Negative power of block encoded matrix, \cite{cgj18}]%\label{thm:bc_invert}
Let $c\in (0,\infty)$ and $\kappa \geq 2$. Let $A$ be a Hermitian matrix such that $\kappa^{-1}I \preceq A \preceq I$. Suppose that $U$ is an $(\alpha, a, \delta)$-block encoding of $A$ that can be implemented using $T_U$ gates. Then, for any $\epsilon>0$, if $\delta = o(\epsilon \kappa^{-(1+c)}\log^{-3}(\kappa/\epsilon))$, we can implement a unitary $U_{A^{-c}}$ that is a $(2\kappa^{c}, a, \epsilon)$-block encoding of $A^{-c}$ in cost $O(\alpha \kappa (T_U + a) \log^2(\kappa/\epsilon))$.  
\end{theorem}
\fi

%\newpage
\section{Quantum second-order SDP solver}

We propose a fast quantum second-order SDP solver in this section. The detailed algorithm is given in Algorithm~\ref{alg:qsdp_fast}. We follow the framework of the robust interior-point method (Algorihtm~\ref{alg:robust_ipm}). In Section \ref{sec:q_slack}, we show the quantum implementation of the procedure \textsc{ApproxSlack} that directly computes the slack matrix inverse. In Section~\ref{sec:q_grad}, we show how to implement the procedure \textsc{ApproxGradient} that computes the gradient in quantum. In Section~\ref{sec:q_update}, we show how to implement the procedure \textsc{ApproxDelta} that computes the Newton step in quantum.  Then, in Section~\ref{sec:q_combine}, we combine them together and show that the robustness conditions are satisfied, which gives the running time and correctness guarantees of Algorithm~\ref{alg:qsdp_fast}.

\begin{algorithm}[!ht]
\caption{Quantum SDP Solver. Let $A_1, \cdots, A_m$ denote a list of $n \times n$ matrices. The input $\mathsf{A}$ has size $m \times n^2$. The input vector $b$ has length $m$. And the input matrix $C$ has size $n$ by $n$. }
\label{alg:qsdp_fast}
\begin{algorithmic}[1]
\Procedure{\textsc{SolveSDPinQuantum}}{$ \{ A_i \}_{i=1}^m$, $b , C, m, n$} 
\State \Comment{The input $\mathsf{A}, b, C$ are stored in QRAM.}
\State \Comment{Initialization} 
\State Choose $\eta$ and $T$ as Definition~\ref{def:eta_T}. %$\eta \leftarrow \frac{1}{n+2}$, ~~$T \leftarrow \frac{40}{\epsilon_N} \sqrt{n} \log (\frac{n}{\epsilon})$
\State Find initial feasible dual vector $y$ and store in QRAM\Comment{ \cite{hjstz22} %Lemma~\ref{lem:initialization}, 
$O(m+n)$-time}\label{ln:q_initialization_new}
\State Update QRAM for the modified SDP instance \Comment{$\wt{O}(m+n)$-time}

\For {$t = 1 \to T$} %\Comment{ Iterations of approximate barrier method}
    \State \Comment{\textcolor{blue}{\textsc{ApproxSlack}}, Lemma~\ref{lem:qspeed_slack_inverse}}\label{ln:qslack_new}
    \State $U_S\gets$ the block-encoding for $S=\sum_{i=1}^m y_i A_i -C$
    \For{$i=1\to n$}\Comment{$\wt{O}(n^{2.5})$-time in total}
        \State $\ket{(\wt{S}^{-1})_i}\gets$ apply quantum linear system solver with $U_S$ and $\ket{e_i}$
        \State $v_i\gets$ apply $\ell_2$-tomography on the state $\ket{(\wt{S^{-1}})_i}$\Comment{$\wt{O}(n^{1.5})$-time}
        \State $N_i\gets$ estimate the norm $\|(S^{-1})_i\|_2$\Comment{$\wt{O}(n^{0.5})$-time}
        \State $(\wt{S}^{-1})_i \gets N_i \cdot v_i$ \Comment{$O(n)$-time}
    \EndFor
    %\State $\ket{\mathsf{A}^\top y}\gets$ apply block encoding of $\mathsf{A}$ to $\ket{y}$\Comment{$\wt{O}(1)$ \textbf{Q}-time}
    %\State $v_S\gets$ apply $\ell_2$-tomography on the state $\ket{\mathsf{A}^\top y}$\Comment{$v_S\in \R^{n^2}$, $\wt{O}(n^2)$ \textbf{Q}-time}
    %\State $\ell_S\gets $ estimate the norm $\|\mathsf{A}^\top y\|_2$\Comment{$\wt{O}(1)$ \textbf{Q}-time}
    %\State $\wt{S}\gets \mathrm{mat}(\ell_S\cdot v_S)-C$ \Comment{$O(n^2)$ \textbf{C}-time}\\
    \State $\wt{S}^{-1}\gets \frac{1}{2}(\wt{S}^{-1}+(\wt{S}^{-1})^\top)$ and store it in QRAM\Comment{$O(n^2)$-time}
	\State $\eta^{\new} \leftarrow \eta \cdot (1 + \frac{\epsilon_N}{20 \sqrt{n}})$
	%\State $\wt{S}^{-1}\gets $invert $\wt{S}$ classically and store it in QRAM \Comment{$O(n^\omega)$ \textbf{C}-time, $\wt{O}(n^2)$ \textbf{Q}-time}\label{ln:compute_S_inv_new}\\
	\State \Comment{\textcolor{blue}{\textsc{ApproxGradient}}, Lemma~\ref{lem:qspeed_grad_state}}\label{ln:qgrad_new}
    \State $U_{\sf A'}\gets$ the block-encoding of $\begin{bmatrix}{\sf A} & b\end{bmatrix}$
    \State $\ket{s'}_\gets$ the vector state for $\begin{bmatrix}-\mathrm{vec}(\wt{S}^{-1}) & \eta\end{bmatrix}^\top$
	\State $\ket{\wt{g}_{\eta^{\new}}}\gets $ apply $U_{\mathsf{A}'}$ to $\ket{s'}$ %\Comment{$\wt{O}(n)$-time}
	%\State $v_g\gets $ apply $\ell_2$-tomography on the state $\ket{\wt{g}_{\eta^{\new}}}$\Comment{$v_g\in \R^{m}$, $\wt{O}(mn)$-time}
	\State $N_g\gets $ estimate the norm $\|{\sf A}' s'\|_2$ \Comment{$\wt{O}(n^2)$-time}
	%\State $\wt{g}_{\eta^{\new}}\gets \eta^{\new}b - \ell_g \cdot v_g$ and store it in QRAM\Comment{$O(m)$ \textbf{C}-time}\\
	\State \Comment {\textcolor{blue}{\textsc{ApproxDelta}}, Lemma~\ref{lem:qspeed_update}}\label{ln:compute_delta_y_new}
	%\State \Comment {Note that matrix $\wt{H}$ is the inverse of $(\mathsf{A} \cdot (\wt{S}^{-1} \otimes \wt{S}^{-1}) \cdot \mathsf{A}^{\top}) $}
	\State $\wt{S}^{-1/2}\gets$ classically compute $\sqrt{\wt{S}^{-1}}$ and store in QRAM \Comment{$O(n^\omega)$-time}
    \State $U_1\gets$ the block-encoding of $\wt{S}^{-1/2}\otimes \wt{S}^{-1/2}$
    \State $U_2\gets$ the block-encoding of $W={\sf A}(\wt{S}^{-1/2}\otimes \wt{S}^{-1/2})$ using $U_{\sf A}$ and $U_1$
    \State $U_H\gets$ the block-encoding of $\wt{H}=WW^\top$ using $U_2$
    \State $\ket{\wt{H}^{-1}\ket{\wt{g}_{\eta^{\new}}}}\gets$ apply quantum linear system solver with $U_{H}$ and $\ket{\wt{g}_{\eta^{\new}}}$ 
	\State $v_\delta\gets $ apply $\ell_2$-tomography on the state $\ket{\wt{H}^{-1}\ket{\wt{g}_{\eta^{\new}}}}$\Comment{$\wt{O}(mn)$-time}
	\State $N_\delta \gets $ estimate the norm $\|\wt{H}^{-1}\ket{\wt{g}_{\eta^{\new}}}\|_2$\Comment{$\wt{O}(n)$-time}
	\State $\wt{\delta}_y\gets -N_\delta\cdot N_g\cdot  v_\delta$ \Comment{$O(m)$-time}\\
	\State $y^{\new} \leftarrow y + \wt{\delta}_y$ and store in QRAM \Comment{$O(m)$-Time}\label{ln:q_compute_y_new}
\EndFor
\EndProcedure
\end{algorithmic}
\end{algorithm}

\subsection{Slack matrix}\label{sec:q_slack}
The goal of this section is to prove the following lemma, which shows how to efficiently compute the slack matrix inverse $S^{-1}$ using quantum singular value transformation (QSVT). We will show in Section~\ref{sec:q_combine} that the error guarantee in the Frobenius norm suffices for our robust IPM framework.
\begin{lemma}[Quantum speedup for computing the slack matrix inverse $S^{-1}$]\label{lem:qspeed_slack_inverse}
Let $S \in \R^{n \times n}$ be defined as $S=\sum_{i=1}^m y_i A_i - C$. Let $\epsilon_S \in (0,1/10)$ denote an accuracy parameter. 
For the \textsc{ApproxSlack} procedure (Line~\ref{ln:qslack_new}), there is an algorithm that runs in time
\begin{align*}
    \wt{O}(n^2 \mu(S)\kappa(S)\epsilon_S^{-1})
\end{align*}
and outputs a classical symmetric matrix $\wt{S}^{-1}\in \R^{n\times n}$ such that 
\begin{align*}
\|\wt{S}^{-1}-S^{-1}\|_F \leq \epsilon_S \cdot \|S^{-1}\|_F.
\end{align*}
%\Zhao{I guess an alternative bound i $\epsilon_S^{-2} \frac{\|S\|_F^2}{\| S \|^2}$ time with $\epsilon_S \| S \|$ error. This bound is scale-invariant and similarly to Ewin Tang's papers.}
%\Ruizhe{$\epsilon_S\in [0, 1]$ is tunable.}
\end{lemma}
\begin{proof}
Using QRAM, we can construct $P_L,P_R$, a $(1+\|y\|_1,\log m, 0)$-state-preparation pair for $\begin{bmatrix}y\\ -1\end{bmatrix}\in \R^{m+1}$. Since $A_1,\dots,A_m$ and $C$ are stored in QRAM, we can implement $U_S$, which is a $(O(\mu), O(\log n), 0)$-block encoding for $S=\sum_{i=1}^m y_i A_i - C$, where $\mu:=\|y\|_1\cdot \max\{\max_{i\in [n]}\mu(A_i), \mu(C)\}$. Then, for each basis state $\ket{e_i}$, we compute $S^{-1}\ket{e_i}$ using Theorem~\ref{thm:qlsa}. More specifically, we can prepare a state $\ket{(\wt{S}^{-1})_i}$ that is $\epsilon_1$-close to $\ket{(S^{-1})_i}$ in time $\wt{O}(\mu\cdot  \kappa(S))$. And we can estimate the norm $N_i\in (1\pm \epsilon_2)\|S^{-1}\ket{e_i}\|_2$ in time $\wt{O}(\mu\cdot  \kappa(S)\cdot \epsilon_2^{-1})$. Finally, by Theorem~\ref{thm:tomography}, it takes $\wt{O}(n\cdot \mu \kappa(S)\cdot \epsilon_3^{-1})$-time to get a classical vector $v_i$ such that $\|v_i-\ket{(\wt{S}^{-1})_i}\|_2 \leq \epsilon_3$. Define $(\wt{S}^{-1})_i := N_i \cdot v_i\in \R^n$. 

We have
\begin{align*}
    \|(\wt{S}^{-1})_i - (S^{-1})_i\|_2= &~ \|N_i \cdot v_i - (S^{-1})_i\|_2\\
    \leq &~ \|(S^{-1})_i\|_2\cdot \|(1\pm \epsilon_2)v_i - \ket{(S^{-1})_i}\|_2\\
    \leq &~ \|(S^{-1})_i\|_2\cdot (\|(1\pm \epsilon_2)v_i - \ket{(\wt{S}^{-1})_i}\|_2 + \|\ket{(\wt{S}^{-1})_i} - \ket{(S^{-1})_i}\|_2)\\
    \leq &~ \|(S^{-1})_i\|_2\cdot (\|(1\pm \epsilon_2)v_i - \ket{(\wt{S}^{-1})_i}\|_2 + \epsilon_1)\\
    \leq &~ \|(S^{-1})_i\|_2\cdot (\|v_i - \ket{(\wt{S}^{-1})_i}\|_2 + \epsilon_1 + \epsilon_2 \|v_i\|_2)\\
    \leq &~ \|(S^{-1})_i\|_2\cdot (\epsilon_1 + \epsilon_3 +  \epsilon_2 \|v_i\|_2)\\
    \leq &~ \|(S^{-1})_i\|_2\cdot (\epsilon_1 + \epsilon_3 +  \epsilon_2 (\|\ket{(\wt{S}^{-1})_i}\|_2 + \|\ket{(\wt{S}^{-1})_i} - v_i\|_2))\\
    \leq &~ \|(S^{-1})_i\|_2\cdot (\epsilon_1 + \epsilon_3 +  \epsilon_2 (1+\epsilon_3))\\
    \leq &~ O(\epsilon_1+\epsilon_2+\epsilon_3)\cdot \|(S^{-1})_i\|_2.
\end{align*}
By taking $\epsilon_1=\epsilon_2=\epsilon_3=O(\epsilon_S)$, we get that $(\wt{S}^{-1})_i$ can be computed in $\wt{O}(n\mu(S)\kappa(S)\epsilon_S^{-1})$-time such that $\|(\wt{S}^{-1})_i-(S^{-1})_i\|_2\leq \epsilon_S\|(S^{-1})_i\|_2$.
We apply this procedure for all $i\in [n]$. Then, we obtain $\wt{S}^{-1}\in \R^{n\times n}$ such that 
\begin{align*}
\|\wt{S}^{-1}-S^{-1}\|_F\leq \epsilon_S\|S^{-1}\|_F
\end{align*}
in $\wt{O}(n^2\mu(S)\kappa(S)\epsilon_S^{-1})$-time.  Since we know that $S$ and $S^{-1}$ are symmetric matrices, we can easily symmetrize $\wt{S}^{-1}$ by $(\wt{S}^{-1}+(\wt{S}^{-1})^\top)/2$ without increasing the approximation error.
\end{proof}

\subsection{Gradient}\label{sec:q_grad}
The goal of this section is to prove the following lemma, which computes a quantum state that encodes the gradient $g_\eta$. We will use the quantum state and the estimated norm to quantumly compute the Newton step later. Furthermore, we also bound the approximation error between the true gradient $g_{\eta, S}$ and the vector encoded in the quantum state, which will be useful in the latter error analysis.
\begin{lemma}[Computing the gradient state $\ket{g_\eta}$ and norm $\|g_\eta\|$]\label{lem:qspeed_grad_state}
Let $g_{\eta} \in \R^m$ be defined as $g_{\eta}=\eta\cdot b - \mathsf{A}\cdot \mathrm{vec}(\wt{S}^{-1}):=\mathsf{A'}\cdot s' \in \R^m$, where $\mathsf{A'}:=\begin{bmatrix}{\sf A}&b\end{bmatrix}$ and $s':=\begin{bmatrix}-\mathrm{vec}(\wt{S}^{-1})\\\eta\end{bmatrix}$. Then, a quantum state $\ket{\wt{g}_\eta}$ can be prepared in time $T_g=\wt{O}(\mu({\sf A})\kappa({\sf A}))$ that is $\epsilon_g$-close to the state $\ket{g_\eta}$. Moreover, $\|g_\eta\|_2$ can be estimated with $\epsilon_g'$-multiplicative error in time $T_{g,{\sf norm}}=\wt{O}(\mu({\sf A})\kappa({\sf A})\epsilon_g'^{-1})$.

Let $g_{\eta,S}=\eta\cdot b - {\sf A}\cdot \mathrm{vec}(S^{-1})$. Then, we have
\begin{align*}
    \|\wt{g}_\eta - g_{\eta,S}\|_2\leq O(\epsilon_g+\epsilon_g' + \epsilon_S\cdot \kappa({\sf A})) \|g_{\eta, S}\|_2.
\end{align*}
\end{lemma}
\begin{proof}
Since ${\sf A}$ and $b$ are stored in QRAM, we can form the block-encoding of ${\sf A}$. Moreover, by Lemma~\ref{lem:qspeed_slack_inverse}, $\wt{S}^{-1}$ are obtained in the classical form such that $\|\wt{S}^{-1}-S^{-1}\|_F\leq \epsilon_S\|S^{-1}\|_F$. We can prepare the state $|s'\rangle$. By Theorem~\ref{thm:qlsa} (part 2), $\ket{\wt{g}_\eta}$ can be prepared in time $\wt{O}(\mu({\sf A}')\kappa({\sf A}'))=\wt{O}(\mu({\sf A})\kappa({\sf A}))$ that is $\epsilon_g$-close to $\ket{g_\eta}$. Furthermore, by Theorem~\ref{thm:qlsa} (part 3), the norm $\|g_\eta\|_2$ can be estimated within relative error $\epsilon_g'$ in time $\wt{O}(\mu({\sf A})\kappa({\sf A})\epsilon_g'^{-1})$. We analyze the approximation error below.

Let $\wt{g}_{\eta}:=N_g\cdot \ket{\wt{g}_{\eta}}$, where $N_g\in (1\pm \epsilon_g')\|g_{\eta}\|_2$ and $\|\ket{\wt{g}_{\eta}} - \ket{g_{\eta}}\|\leq \epsilon_g$. Then, we have
\begin{align}\label{eq:diff_wt_g_and_g}
    \|\wt{g}_{\eta} - g_{\eta}\|_2=&~ \|N_g\cdot \ket{\wt{g}_{\eta}} - \|g_{\eta}\|_2\cdot \ket{g_{\eta}}\|_2\notag\\
    \leq &~ \|N_g\cdot \ket{\wt{g}_{\eta}} - N_g\cdot \ket{g_{\eta}}\|_2 + \|N_g\cdot \ket{g_{\eta}} - \|g_{\eta}\|_2\cdot \ket{g_{\eta}}\|_2\notag\\
    \leq &~ \epsilon_g N_g + \|N_g\cdot \ket{g_{\eta}} - \|g_{\eta}\|_2\cdot \ket{g_{\eta}}\|_2\notag\\
    = &~ \epsilon_g N_g + |N_g - \|g_{\eta}\|_2|\notag\\
    \leq &~ (\epsilon_g(1+\epsilon_g') + \epsilon_g')\|g_{\eta}\|_2\notag\\
    = &~ O(\epsilon_g+\epsilon_g')\|g_{\eta}\|_2,
\end{align}    
where the second step follows from triangle inequality, the third step follows from the definition of $\epsilon_g$, the fourth step follows from $\|\ket{g_\eta}\|_2=1$, the fifth step follows from the definition of $\epsilon_g'$, and the last step is straightforward.

Define $g_{\eta,S}=\eta\cdot b - {\sf A}\cdot \mathrm{vec}(S^{-1})$. We have
\begin{align*}
    \|g_\eta - g_{\eta, S}\|_2 = &~ \|{\sf A}\cdot (\mathrm{vec}(\wt{S}^{-1}) - \mathrm{vec}(S^{-1}))\|_2\\
    \leq &~ \|{\sf A}\|\cdot \|\mathrm{vec}(\wt{S}^{-1}) - \mathrm{vec}(S^{-1})\|_2\\
    = &~ \|{\sf A}\|\cdot \|\wt{S}^{-1} - S^{-1}\|_F\\
    \leq &~ \|{\sf A}\|\cdot \epsilon_S\|S^{-1}\|_F,
\end{align*}
where the second step follows from the definition of the spectral norm, and the third step follows from Lemma~\ref{lem:qspeed_slack_inverse}.

\iffalse
Note that 
\begin{align}\label{eq:diff_S_F_norm}
    \|\wt{S}^{-1} - S^{-1}\|_F
    \leq &~ \|\wt{S}^{-1}\| \cdot \|S^{-1}\|\cdot \|\wt{S}-S\|_F\notag\\
    \leq &~\|\wt{S}^{-1}\| \cdot \|S^{-1}\|\cdot \epsilon_S\|S^{-1}\|^{-1}\notag\\
    = &~ \epsilon_S \cdot \|\wt{S}^{-1}\|,
\end{align}
where the first step follows from Theorem~\ref{thm:wedin_frobenius}, the second step follows from Lemma~\ref{lem:qspeed_slack}.
\fi
Hence, we get that
\iffalse
\begin{align}\label{eq:diff_g_eta_S}
    \frac{\|g_\eta - g_{\eta, S}\|_2}{\|g_{\eta, S}\|_2} \leq &~ \frac{\epsilon_S \|{\sf A}\|\cdot \|\wt{S}^{-1}\|}{\|\eta\cdot b - {\sf A}\cdot \mathrm{vec}(S^{-1})\|_2}\notag\\
    = &~ O(\epsilon_S)\cdot \frac{\|{\sf A}\|\cdot \|\wt{S}^{-1}\|}{\|{\sf A}\cdot \mathrm{vec}(S^{-1})\|_2}\notag\\
    \leq &~ O(\epsilon_S)\cdot\kappa({\sf A}) \cdot \frac{\|\wt{S}^{-1}\|}{\|S^{-1}\|_F}\notag\\
    \leq &~ O(\epsilon_S)\cdot\kappa({\sf A}) \cdot \frac{\|\wt{S}^{-1}\|}{\|\wt{S}^{-1}\|_F - \|\wt{S}^{-1}-S^{-1}\|_F}\notag\\
    \leq &~ O(\epsilon_S)\cdot\kappa({\sf A}) \cdot \frac{\|\wt{S}^{-1}\|}{(1-\epsilon_S)\|\wt{S}^{-1}\|_F}\notag\\
    = &~ O(\epsilon_S \cdot \kappa({\sf A})),
\end{align}
where the second step follows from the fact that $\|b\|_2$ can be re-scaled to be $n^{-O(1)}$ so that $\|\eta b\|_2\ll \|{\sf A}\cdot \mathrm{vec}(S^{-1})\|_2$, the third step follows from $\|{\sf A}\cdot \mathrm{vec}(S^{-1})\|_2\geq \sigma_{\min}({\sf A})\|\mathrm{vec}(S^{-1})\|_2=\sigma_{\min}({\sf A})\|S^{-1}\|_F$, the fourth step follows from triangle inequality, the fifth step follows from Eq.~\eqref{eq:diff_S_F_norm}, and the last step is straightforward.  
\fi
\begin{align}\label{eq:diff_g_eta_S}
    \frac{\|g_\eta - g_{\eta, S}\|_2}{\|g_{\eta, S}\|_2} \leq &~ \frac{\epsilon_S \|{\sf A}\|\cdot \|S^{-1}\|_F}{\|\eta\cdot b - {\sf A}\cdot \mathrm{vec}(S^{-1})\|_2}\notag\\
    = &~ O(\epsilon_S)\cdot \frac{\|{\sf A}\|\cdot \|\wt{S}^{-1}\|}{\|{\sf A}\cdot \mathrm{vec}(S^{-1})\|_2}\notag\\
    \leq &~ O(\epsilon_S)\cdot\kappa({\sf A}) \cdot \frac{\|S^{-1}\|_F}{\|S^{-1}\|_F}\notag\\
    %\leq &~ O(\epsilon_S)\cdot\kappa({\sf A}) \cdot \frac{\|\wt{S}^{-1}\|}{\|\wt{S}^{-1}\|_F - \|\wt{S}^{-1}-S^{-1}\|_F}\notag\\
    %\leq &~ O(\epsilon_S)\cdot\kappa({\sf A}) \cdot \frac{\|\wt{S}^{-1}\|}{(1-\epsilon_S)\|\wt{S}^{-1}\|_F}\notag\\
    = &~ O(\epsilon_S \cdot \kappa({\sf A})),
\end{align}
where the second step follows from the fact that $\|b\|_2$ can be re-scaled to be $n^{-O(1)}$ so that $\|\eta b\|_2\ll \|{\sf A}\cdot \mathrm{vec}(S^{-1})\|_2$, the third step follows from $\|{\sf A}\cdot \mathrm{vec}(S^{-1})\|_2\geq \sigma_{\min}({\sf A})\|\mathrm{vec}(S^{-1})\|_2=\sigma_{\min}({\sf A})\|S^{-1}\|_F$, and the third step is straightforward.

Therefore,
\begin{align*}
    \|\wt{g}_\eta - g_{\eta, S}\|_2 \leq &~ \|\wt{g}_{\eta} - g_{\eta}\|_2 + \|g_\eta - g_{\eta, S}\|_2 \\
    = &~ O(\epsilon_g+\epsilon_g')\|g_{\eta}\|_2 + O(\epsilon_S \cdot \kappa({\sf A})) \|g_{\eta, S}\|_2\\
    \leq &~ O((\epsilon_g+\epsilon_g')(1+\epsilon_S\cdot \kappa({\sf A})) + \epsilon_S\cdot \kappa({\sf A})) \|g_{\eta, S}\|_2\\
    = &~ O(\epsilon_g+\epsilon_g' + \epsilon_S\cdot \kappa({\sf A})) \|g_{\eta, S}\|_2.
\end{align*}
where the first step follows from triangle inequality, the second step follows from Eqs.~\eqref{eq:diff_wt_g_and_g} and \eqref{eq:diff_g_eta_S}, the third step follows from Eq.~\eqref{eq:diff_g_eta_S} again, and the last step follows directly.
\end{proof}

\subsection{Update the changes of the dual}\label{sec:q_update}
The goal of this section is to prove the following lemma that computes a classical vector that is close to the Newton step $\delta_y$, based on QSVT and tomography. We also give an $\ell_2$-relative error guarantee.
\begin{lemma}[Quantum speedup for computing the update $\delta_y$]\label{lem:qspeed_update}
For the \textsc{ApproxDelta} procedure (Line~\ref{ln:compute_delta_y_new}), there is an algorithm that outputs a classical vector $\wt{\delta}_y\in \R^m$ such that
\begin{align*}
    \|\wt{\delta}_y-(-\wt{H}^{-1}\wt{g}_{\eta^{\new}})\|_2\leq O(\epsilon_\delta + \epsilon_n + \epsilon_\delta')\cdot \|\wt{H}^{-1}\wt{g}_{\eta^{\new}}\|_2
\end{align*}
in time 
\begin{align*}
    \wt{O}\Big(\left((m\epsilon_n^{-1} + \epsilon_\delta')\|\wt{H}\| + \epsilon_g'^{-1}\right)\mu({\sf A})\kappa({\sf A})\kappa(\wt{H})+ (m\epsilon_n^{-1} + \epsilon_\delta'^{-1})\|\wt{H}\|\mu(\wt{S}^{-1})\kappa(\wt{S})\kappa(\wt{H}) \Big),
\end{align*}
where $\wt{H}=\mathsf{A}(\wt{S}^{-1}\otimes \wt{S}^{-1})\mathsf{A}^\top$. 
\end{lemma}

\begin{proof}
We have access to $\wt{S}^{-1}$ in QRAM. Then, we can compute $(\wt{S}^{-1})^{1/2}$ classically and construct a unitary $U_{\wt{S}^{-1/2}}$, which is a $(\mu(\wt{S}^{-1/2}), O(\log n), 0)$-block encoding of $(\wt{S}^{-1})^{1/2}$. 
\iffalse
From the proof of Lemma~\ref{lem:qspeed_grad}, we know that
\begin{align*}
    \|\wt{S}^{-1}-S^{-1}\|\leq O(\epsilon_S)\max\{\|\wt{S}^{-1}\|^2, \|S^{-1}\|^2\}\|S+C\|~=:\epsilon_1.
\end{align*}
Hence, $U_{\wt{S}^{-1}}$ is also a $(\mu(\wt{S}^{-1}), O(\log n), \epsilon_1)$-block encoding of $S^{-1}$.
\fi
%By Theorem~\ref{thm:be_construct_kronecker}, we can implement a unitary $U_{\otimes}$ by applying $U_{\wt{S}^{-1}}$ twice such that it is a $(\mu(\wt{S}^{-1})^2, O(\log n), 0)$-block encoding of $\wt{S}^{-1}\otimes \wt{S}^{-1}$.

%By Theorem~\ref{thm:be_product}, we can implement a unitary $U_{\wt{H}}$ that is a $(\mu(\mathsf{A})^2\mu(\wt{S}^{-1})^2, O(\log n), 0)$-block encoding of $\mathsf{A}(\wt{S}^{-1}\otimes \wt{S}^{-1})\mathsf{A}^\top$.

%Note that 
%\begin{align*}
%    \wt{H}^{-1}\wt{g}_{\eta^{\new}}=\wt{H}^{-1} (\eta^{\new} b - \mathsf{A}\cdot \mathrm{vec}(\wt{S}^{-1})) = \eta^{\new} \wt{H}^{-1} b - \wt{H}^{-1} \mathsf{A}\cdot \mathrm{vec}(\wt{S}^{-1}).
%\end{align*}

By Theorem~\ref{thm:be_construct_kronecker}, $U_1$ implements a $(\mu(\wt{S}^{-1}), O(\log n), 0)$-block encoding of $\wt{S}^{-1/2}\otimes \wt{S}^{-1/2}$ in time $T_1=\wt{O}(1)$.

Then, by Lemma~\ref{lem:pre_amp_prod}, $U_2$ implements a $(O(\|{\sf A}\|\|\wt{S}^{-1}\|), O(\log n), \delta_2)$-block encoding of ${\sf A}(\wt{S}^{-1/2}\otimes \wt{S}^{-1/2})$ in time
\begin{align*}
    T_2=&~\wt{O}_{\delta_2}\Big(\mu({\sf A})/\|{\sf A}\| \cdot (T_A+\log n)+ \mu(\wt{S}^{-1})/\|\wt{S}^{-1}\| \cdot (T_1+\log n)\Big)\\
    = &~ \wt{O}_{n,\delta_2}\Big(\mu({\sf A})/\|{\sf A}\|+ \mu(\wt{S}^{-1})/\|\wt{S}^{-1}\|\Big).
\end{align*}

Since $\wt{H} = ({\sf A}(\wt{S}^{-1/2}\otimes \wt{S}^{-1/2})) \cdot ({\sf A}(\wt{S}^{-1/2}\otimes \wt{S}^{-1/2}))^\top$, $U_H$ implements a 
$(O(\|\wt{H}\|), O(\log n), O(\delta_2))$-block encoding of $\wt{H}$ in time
\begin{align*}
    T_H = &~ \wt{O}_{n,\delta_2}\Big(\|{\sf A}\|\|\wt{S}^{-1}\| / \|{\sf A}(\wt{S}^{-1/2}\otimes \wt{S}^{-1/2})\| \cdot T_2\Big)\\
    = &~  \wt{O}_{n,\delta_2}\Big(\left(\mu({\sf A})\|\wt{S}^{-1}\|+ \mu(\wt{S}^{-1})\|{\sf A}\|\right)\|\wt{H}\|^{-1/2}\Big).
\end{align*}

Then, we can prepare a quantum state $\ket{\wt{\delta}_y}$ that is $\epsilon_\delta$-close to $\ket{\wt{H}^{-1}\wt{g}_\eta}$ in time
\begin{align*}
    T_\delta = &~ \wt{O}\Big(\kappa(\wt{H})(T_H \|\wt{H}\| + T_g)\Big)\\
    = &~\wt{O}\Big(\kappa(\wt{H})\Big((\mu({\sf A})\|\wt{S}^{-1}\|+ \mu(\wt{S}^{-1})\|{\sf A}\|)\|\wt{H}\|^{-1/2} \cdot \|\wt{H}\| + \mu({\sf A})\kappa({\sf A})\Big)\Big)\\
    = &~ \wt{O}\Big(\kappa(\wt{H})\Big((\mu({\sf A})\|\wt{S}^{-1}\|+ \mu(\wt{S}^{-1})\|{\sf A}\|)\|\wt{H}\|^{1/2} + \mu({\sf A})\kappa({\sf A})\Big)\Big).
\end{align*}
And the norm $\|\wt{H}^{-1}\ket{\wt{g}_\eta}\|_2$ can be estimated with $\epsilon_\delta'$-multiplicative error in time $T_{\delta,{\sf norm}}=\wt{O}(T_\delta\cdot \epsilon_\delta'^{-1})$. 

Finally, by the tomography procedure, we can obtain a classical vector $\wt{\delta}_{y,1}\in \R^m$ such that $\|\ket{\wt{\delta}_y}-\wt{\delta}_{y,1}\|_2\leq \epsilon_n$ in time $T_{\delta, {\sf tomo}}=\wt{O}(T_\delta \cdot m \epsilon_n^{-1})$. And we let $\wt{\delta}_y:=N_\delta \cdot N_g \cdot \wt{\delta}_{y,1}$, where $N_\delta \in (1\pm \epsilon_\delta')\|\wt{H}^{-1}\ket{\wt{g}_\eta}\|_2$ and $N_g\in (1+\epsilon_g')\|g_\eta\|_2$.

We have
\begin{align*}
    \|\wt{H}^{-1}\wt{g}_\eta - \wt{\delta}_y\|_2 = &~ \| \wt{H}^{-1}\wt{g}_\eta - N_\delta \cdot N_g \cdot \wt{\delta}_{y,1}\|_2\\
    = &~ \|\wt{g}_\eta\|_2 \cdot \| \wt{H}^{-1}\ket{\wt{g}_\eta} - N_\delta  \cdot \wt{\delta}_{y,1}\|_2\\
    = &~ \|\wt{g}_\eta\|_2 \cdot \|\wt{H}^{-1}\ket{\wt{g}_\eta}\|_2 \cdot \| \ket{\wt{H}^{-1}\wt{g}_\eta} - (1\pm \epsilon_\delta') \wt{\delta}_{y,1}\|_2\\
    = &~ \|\wt{H}^{-1}g_\eta\|_2\cdot (\|\ket{\wt{H}^{-1}g_\eta} - \wt{\delta}_{y,1}\|_2 + \epsilon_\delta'\|\wt{\delta}_{y,1}\|_2)\\
    = &~ \|\wt{H}^{-1}\wt{g}_\eta\|_2\cdot (\|\ket{\wt{H}^{-1}\wt{g}_\eta} - \wt{\delta}_{y,1}\|_2 + O(\epsilon_\delta'))\\
    = &~ \|\wt{H}^{-1}\wt{g}_\eta\|_2\cdot (\|\ket{\wt{H}^{-1}\wt{g}_\eta} - \ket{\wt{\delta}_y}\|_2 + \|\ket{\wt{\delta}_y} - \wt{\delta}_{y,1}\|_2 + O(\epsilon_\delta')) \\
    = &~ \|\wt{H}^{-1}\wt{g}_\eta\|_2\cdot O(\epsilon_\delta + \epsilon_n + \epsilon_\delta'),
\end{align*}
where the first step follows from the definition of $\wt{\delta}_y$, the second step follows from $\|\wt{g}_\eta\|_2=N_g$, the third step follows from $N_\delta \in (1\pm \epsilon_\delta')\|\wt{H}^{-1}\ket{\wt{g}_\eta}\|_2$, the forth step follows from triangle inequality, the fifth step follows from $\|\wt{\delta}_{y,1}\|_2=O(1)$, the sixth step follows from triangle inequality, and the last step follows from the approximation guarantees of $\|\ket{\wt{H}^{-1}\wt{g}_\eta}-\ket{\wt{\delta}_g}\|_2\leq \epsilon_\delta$ and $\|\ket{\wt{\delta}_y}-\wt{\delta}_{y,1}\|_2\leq \epsilon_n$.

The total time complexity of getting $\wt{\delta}_y$ is:
\begin{align*}
    T=&~ T_{\delta, {\sf tomo}} + T_{\delta, {\sf norm}} + T_{g, {\sf norm}}\\
    = &~ \wt{O}((m\epsilon_n^{-1} + \epsilon_\delta'^{-1}) T_\delta + \epsilon_g'^{-1}\mu({\sf A})\kappa({\sf A}))\\
    = &~ \wt{O}\Big((m\epsilon_n^{-1} + \epsilon_\delta'^{-1})\kappa(\wt{H})(\mu({\sf A})\|\wt{S}^{-1}\|+ \mu(\wt{S}^{-1})\|{\sf A}\|)\|\wt{H}\|^{1/2} + \epsilon_g'^{-1}\kappa(\wt{H})\mu({\sf A})\kappa({\sf A})\Big)\\
    = &~ \wt{O}\Big(\left((m\epsilon_n^{-1} + \epsilon_\delta')\|\wt{H}\| + \epsilon_g'^{-1}\right)\mu({\sf A})\kappa({\sf A})\kappa(\wt{H})+ (m\epsilon_n^{-1} + \epsilon_\delta'^{-1})\|\wt{H}\|\mu(\wt{S}^{-1})\kappa(\wt{S})\kappa(\wt{H}) \Big).
\end{align*}

\end{proof}

\subsection{Combine}\label{sec:q_combine}
In this section, we wrap up the three components and show the cost-per-iteration of our quantum SDP solver in the following lemma. By our error analysis of each component, we can prove that the robustness conditions are satisfied by proper choices of parameters. Then, the main theorem (Theorem~\ref{thm:quantum_main}) follows immediately.
\begin{lemma}[Quantum SDP solver cost per iteration]\label{lem:q_cost_per_iter}
Let $t\in [T]$. The $t$-th iteration of Algorithm~\ref{alg:qsdp_fast} takes
\begin{align*}
   \wt{O}((m\mu({\sf A})+n^2\mu(S)) \cdot \kappa({A})\kappa(S)\kappa(H)^3)
\end{align*}
time such that \textbf{Cond. 0.} to \textbf{Cond. 4.} in Lemma~\ref{lem:invariant_newton} are satisfied.
\end{lemma}
\begin{proof}
For \textbf{Cond. 0.}, if $t=1$, then a feasible solution $y$ can be found in the initialization step (Line~\ref{ln:q_initialization_new}). For $t>1$, we will prove by induction. Suppose it is satisfied by the $y$ computed in the $(t-1)$-th iteration.

\iffalse
For \textbf{Cond. 1.}, let $\wt{S}$ be the matrix computed by \textsc{ApproxSlack}, and let $S:=\mathrm{mat}(\mathsf{A}^\top y) - C$. Then, by Lemma~\ref{lem:qspeed_slack}, the algorithm takes
\begin{align}\label{eq:qtime_s}
    T_S=\wt{O}(n^2\cdot \kappa(\mathsf{A})\mu(\mathsf{A})\|\mathsf{A}^\top y\|_2\|S^{-1}\| \epsilon_S^{-1})
\end{align}
time such that
\begin{align*}
    \|\wt{S}-S\|_F \leq \epsilon_S \cdot \|S^{-1}\|^{-1}.
\end{align*}
Then, it implies that $\|\wt{S}-S\| \leq \epsilon_S \cdot \|S^{-1}\|^{-1}$. Hence, we have $(1-\epsilon_S)S\preceq \wt{S}\preceq (1+\epsilon_S)S$ by Weyl's inequality. Hence, \textbf{Cond. 1.} is satisfied if 
\begin{align}\label{eq:choose_eps_S}
    \epsilon_S\in (0, 10^{-4}).
\end{align}
\fi
%However, by Lemma~\ref{lem:qspeed_grad}, we need to take $\epsilon_S = \epsilon_g \kappa(\mathsf{A})^{-1}\kappa(H)^{-1}\|b\|_2^{-1}$.
For \textbf{Cond. 1.}, let $\wt{S}^{-1}$ be the matrix computed by \textsc{ApproxSlack}, and let $S:=\mathrm{mat}(\mathsf{A}^\top y) - C$. By Lemma~\ref{lem:qspeed_slack_inverse}, the algorithm takes
\begin{align}\label{eq:qtime_s}
    T_S=\wt{O}(n^2 \mu(S)\kappa(S) \epsilon_S^{-1})
\end{align}
time such that
\begin{align*}
    \|\wt{S}^{-1}-S^{-1}\|_F \leq \epsilon_S \cdot \|S^{-1}\|_F.
\end{align*}
By {\cite[Fact 7.1]{hjstz22}}, we get that $(1-\epsilon_S)S^{-1}\preceq \wt{S}^{-1}\preceq (1+\epsilon_S)S^{-1}$. Hence, $\alpha_S^{-1}S^{-1}\preceq \wt{S}^{-1}\preceq \alpha_S S^{-1}$ holds for some $\alpha_S\in (1,1+10^{-4})$, as long as $\epsilon_S$ is a small enough constant.

For \textbf{Cond. 2.}, since we compute $\wt{S}^{-1}$ and store it in QRAM, we may assume that there is no error in this step, i.e., $\wt{H}=H(\wt{S})$. Hence, this condition is satisfied by taking $\alpha_H=1$.

For \textbf{Cond. 3.}, let $\wt{g}_{\eta^{\new}}:=N_g\cdot \ket{\wt{g}_{\eta^{\new}}}$ produced by Lemma~\ref{lem:qspeed_grad_state}. It guarantees that
\begin{align*}
    \|\wt{g}_{\eta^{\new}} - g_{\eta^{\new}}\|_2\leq O(\epsilon_g+\epsilon_g' + \epsilon_S\cdot \kappa({\sf A})) \|g_{\eta^{\new}}\|_2.
\end{align*}
By the first part of Remark~\ref{rmk:l_2_cond}, if we have
\begin{align*}
    \|\wt{g}_{\eta^{\new}}-g_{\eta^{\new}}\|_2 \leq \tau_g \|g_{\eta^{\new}}\|_2 / \kappa(H).
\end{align*}
for some $\tau_g\in (0,10^{-4})$, then \textbf{Cond. 3.} will be satisfied, i.e.,
\begin{align*}
    \|\wt{g}_{\eta^{\new}} - g_{\eta^{\new}}\|_{H^{-1}} \leq \tau_g \|g_{\eta^{\new}}\|_{H^{-1}}.
\end{align*}
Hence, we can take
\begin{align}\label{eq:choose_eps_g}
    \epsilon_g=\epsilon_g'=O(\kappa(H)^{-1})~~~\text{and}~~~\epsilon_S = O(\kappa({\sf A})^{-1}\kappa(H)^{-1}).
\end{align}

\iffalse
By Lemma~\ref{lem:qspeed_grad}, it takes
\begin{align}\label{eq:qtime_g}
    \wt{O}(m \cdot \kappa(\mathsf{A})^3\kappa(H)^2\|b\|_2^2 \mu(\mathsf{A})\epsilon_g^{-2})
\end{align}
time such that
\begin{align*}
    \|\wt{g}_{\eta^{\new}}-g_{\eta^{\new}}\|_2 \leq \epsilon_g \|g_{\eta^{\new}}\|_2 / \kappa(H). 
\end{align*}

Hence, this condition is satisfied for $\epsilon_g\in (0, 10^{-4})$.
\fi

For \textbf{Cond. 4.}, let $\wt{\delta}_y$ be the vector computed by \textsc{ApproxDelta}. By Lemma~\ref{lem:qspeed_update}, it takes
\begin{align}\label{eq:qtime_y_1}
    T_\delta=\wt{O}\Big(\left((m\epsilon_n^{-1} + \epsilon_\delta')\|\wt{H}\| + \epsilon_g'^{-1}\right)\mu({\sf A})\kappa({\sf A})\kappa(\wt{H})+ (m\epsilon_n^{-1} + \epsilon_\delta'^{-1})\|\wt{H}\|\mu(\wt{S}^{-1})\kappa(\wt{S})\kappa(\wt{H}) \Big)
\end{align}
time such that
\begin{align*}
    \|\wt{\delta}_y-(-\wt{H}^{-1}\wt{g}_{\eta^{\new}})\|_2\leq O(\epsilon_\delta + \epsilon_n + \epsilon_\delta')\cdot \|\wt{H}^{-1}\wt{g}_{\eta^{\new}}\|_2.
\end{align*}
By the second part of Remark~\ref{rmk:l_2_cond}, if we have
\begin{align*}
    \|\wt{\delta}_y-(-\wt{H}^{-1}\wt{g}_{\eta^{\new}})\|_2\leq \tau_\delta\|\wt{H}^{-1}\wt{g}_{\eta^{\new}}\|_2/\kappa(H)
\end{align*}
for some $\tau_\delta\in (0,10^{-4})$, then \textbf{Cond. 4.} will be satisfied, i.e., 
\begin{align*}
    \|\wt{\delta}_y -(- \wt{H}^{-1} \wt{g}_{\eta^{\new}})\|_{H} \leq \tau_{\delta} \cdot \| \wt{H}^{-1} \wt{g}_{\eta^{\new}} \|_{H}.
\end{align*}
Hence, we can take 
\begin{align}\label{eq:choose_eps_delta}
    \epsilon_\delta=\epsilon_n = \epsilon_\delta' = O(\kappa(H)^{-1}).
\end{align}
Then, Eq.~\eqref{eq:qtime_y_1} can be simplified as follows:
\begin{align}\label{eq:qtime_y_2}
    T_\delta = &~ \wt{O}(m(\mu({\sf A})\kappa({\sf A})+\mu(\wt{S}^{-1})\kappa(\wt{S}))\cdot \|H\|\kappa(H)^2)\notag\\
    = &~ \wt{O}(m(\mu({\sf A})\kappa({\sf A})+\mu(S)\kappa(S))\cdot \|H\|\kappa(H)^2).
\end{align}

Then, by Lemma~\ref{lem:invariant_newton}, for $y^\new$ computed in Line~\ref{ln:q_compute_y_new}, we have
\begin{align*}
    \| g (y^{\new}, \eta^{\new} ) \|_{  H(y^{\new})^{-1} } \leq \epsilon_N,
\end{align*}
which means $y^\new$ satisfies the \textbf{Cond. 0.} in the $(t+1)$-th iteration.

Therefore, for all $t\in [T]$, \textbf{Cond. 0.} to \textbf{Cond. 4.} in Lemma~\ref{lem:invariant_newton} are satisfied.

For the running time per iteration, 
%first note that in Line~\ref{ln:compute_S_inv_new}, it takes ${\cal T}_{\text{iter, c}} := O(n^{\omega})$ classical time to compute the inverse matrix. The remaining part of Algorithm~\ref{alg:qsdp_fast} in each iteration will be quantum-classical hybrid, with running time given 
by Eqs.~\eqref{eq:qtime_s} and \eqref{eq:qtime_y_2}:
\iffalse
\begin{align*}
    \mathcal{T}_{\text{iter, hybrid}}:=&~ T_S + T_\delta\\
    =&~ \wt{O}(n^2\cdot \kappa(\mathsf{A})\mu(\mathsf{A})\|\mathsf{A}^\top y\|_2\|S^{-1}\| \epsilon_S^{-1})+\wt{O}(m(\mu({\sf A})\kappa({\sf A})+\mu(\wt{S}^{-1})\kappa(\wt{S}))\cdot \|H\|\kappa(H)^2)\\
    = &~ \wt{O}(m\cdot \kappa(H)^3\mu(\mathsf{A})^2\mu(S^{-1})^2+n^2\cdot \kappa(\mathsf{A})^3\kappa(H)^2\mu(\mathsf{A})\|\mathsf{A}^\top y\|_2^2\|S^{-1}\|^2),
\end{align*}
where the 2nd step is due to $\wt{S}$ is a PSD-approximation of $S$, and $\epsilon_S, \epsilon_g, \epsilon_\delta$ are constants.
\fi
\begin{align*}
    \mathcal{T}_{\text{iter}}:=&~ T_S + T_\delta\\
    =&~ \wt{O}(n^2 \mu(S)\kappa(S) \kappa({\sf A})\kappa(H))+\wt{O}(m(\mu({\sf A})\kappa({\sf A})+\mu(S^{-1})\kappa(S))\cdot \|H\|\kappa(H)^2)\\
    \leq &~ \wt{O}((m\mu({\sf A})+n^2\mu(S)) \cdot \kappa({A})\kappa(S)\kappa(H)^3).
\end{align*}

\iffalse
Therefore, in Algorithm~\ref{alg:qsdp_fast}, each iteration runs in time
\begin{align*}
    \mathcal{T}_{\text{iter}} = &~ {\cal T}_{\text{iter, c}} + \mathcal{T}_{\text{iter, hybrid}}\\
    = &~ \wt{O}(m\cdot \kappa(H)^3\mu(\mathsf{A})^2\mu(S^{-1})^2+n^2\cdot \kappa(\mathsf{A})^3\kappa(H)^2\mu(\mathsf{A})\|\mathsf{A}^\top y\|_2^2\|S^{-1}\|^2+n^\omega).
\end{align*}
\fi
\end{proof}

\begin{theorem}[Quantum second-order SDP solver]\label{thm:quantum_main}
Suppose the input matrices $\{A_i\}_{i=1}^m, C$ and vector $b$ are stored in QRAM. Let $T:=O(\sqrt{n}\log(1/\epsilon))$. Then, the quantum second-order SDP solver (Algorithm~\ref{alg:qsdp_fast}) runs $T$ iterations with 
\begin{align*}
    \wt{O}((m\mu({\sf A})+n^2\mu(S)) \cdot \kappa({A})\kappa(S)\kappa(H)^3)
\end{align*}
time per iteration and outputs a classical vector $y \in \mathbb{R}^{m}$ such that with probability at least $1-1/\poly(n)$, 
\begin{align*}
&\langle b , y  \rangle \leq \langle b , y^* \rangle + \epsilon,~~\text{and}\\
&\sum_{i=1}^n y_i A_i - C \succeq -\epsilon I,
\end{align*}
In the above formula, we use $y^*$ to represent an optimal dual solution.

Furthermore, the algorithm can also output a PSD matrix $X \in \mathbb{R}^{n \times n}_{\geq 0}$ in the same running time that satisfies 
\begin{align}
\langle C, X \rangle \geq \langle C, X^* \rangle - \epsilon \cdot \| C \|_2 \cdot R \quad \text{and} \quad \sum_{i = 1 }^m \left| \langle A_i, \widehat{X} \rangle - b_i \right| \leq 4 n \epsilon \cdot \Big( R  \sum_{i = 1}^m \| A_i \|_1 + \| b \|_1 \Big) .
\end{align}
In the above formula, we use $X^*$ to represent an optimal solution to the semi-definite program in Definition~\ref{def:sdp_primal}. %For matrix $\| $ and $\| A_i \|_1$ is the Schatten $1$-norm of matrix $A_i$.

\end{theorem}

\begin{proof}

Note that the running time per iteration follows from Lemma~\ref{lem:q_cost_per_iter}. %\Zhao{I think this label is a bit wrong before, I fixed.} 
And the number of iterations to achieve the stated error guarantee is by Theorem~\ref{thm:approx_central_path_dual}.

For the primal solution, by Lemma B.1 %~\ref{lem:dual_2_primal} 
in \cite{hjstz22}, we can find $X$ using $y$ classically in time $O(n^\omega)$, which is less than the SDP solver's running time.
\end{proof}

\begin{remark}
Plugging-in the general upper bounds for $\mu({\sf A})\leq n$ and $\mu(S)\leq \sqrt{n}$, we get that the total time complexity of Algorithm~\ref{alg:qsdp_fast} is 
\begin{align*}
    \wt{O}\left(\sqrt{n}(mn+n^{2.5})\poly(\kappa, \log(1/\epsilon))\right).
\end{align*}
\end{remark}

%\newpage
%\onecolumn
%\appendix
%\addcontentsline{toc}{section}{Appendix}
%\section*{Appendix}
%\input{initialization}

%\newpage
%\input{other}
%\newpage
%\input{primal_dual}

%\input{sketch}%%%zhao: I am working on this part now.
\addcontentsline{toc}{section}{References}
\bibliographystyle{alpha}
\bibliography{ref}

\newcommand{\etalchar}[1]{$^{#1}$}
\begin{thebibliography}{vACGN22}

\bibitem[AG19]{ag18}
Joran~van Apeldoorn and Andr{\'a}s Gily{\'e}n.
\newblock Improvements in quantum sdp-solving with applications.
\newblock In {\em ICALP}. arXiv preprint arXiv:1804.05058, 2019.

\bibitem[AGGW17]{aggw17}
Joran~van Apeldoorn, Andr{\'a}s Gily{\'e}n, Sander Gribling, and Ronald~de
  Wolf.
\newblock Quantum sdp-solvers: Better upper and lower bounds.
\newblock In {\em 2017 IEEE 58th Annual Symposium on Foundations of Computer
  Science (FOCS)}, pages 403--414. IEEE, 2017.

\bibitem[AHO98]{aho98}
Farid Alizadeh, Jean-Pierre~A Haeberly, and Michael~L Overton.
\newblock Primal-dual interior-point methods for semidefinite programming:
  convergence rates, stability and numerical results.
\newblock {\em SIAM Journal on Optimization}, 8(3):746--768, 1998.

\bibitem[AK07]{ak07}
Sanjeev Arora and Satyen Kale.
\newblock A combinatorial, primal-dual approach to semidefinite programs.
\newblock In {\em Proceedings of the 39th Annual {ACM} Symposium on Theory of
  Computing (STOC)}, 2007.

\bibitem[ALO16]{alo16}
Zeyuan {Allen Zhu}, Yin~Tat Lee, and Lorenzo Orecchia.
\newblock Using optimization to obtain a width-independent, parallel, simpler,
  and faster positive {SDP} solver.
\newblock In {\em Proceedings of the Twenty-Seventh Annual {ACM-SIAM} Symposium
  on Discrete Algorithms(SODA)}, 2016.

\bibitem[Ans00]{a00}
Kurt~M Anstreicher.
\newblock The volumetric barrier for semidefinite programming.
\newblock {\em Mathematics of Operations Research}, 2000.

\bibitem[ANTZ21]{antz21}
Brandon Augustino, Giacomo Nannicini, Tam{\'a}s Terlaky, and Luis~F Zuluaga.
\newblock Quantum interior point methods for semidefinite optimization.
\newblock {\em arXiv preprint arXiv:2112.06025}, 2021.

\bibitem[AZL17]{al17}
Zeyuan Allen-Zhu and Yuanzhi Li.
\newblock Follow the compressed leader: faster online learning of eigenvectors
  and faster mmwu.
\newblock In {\em Proceedings of the 34th International Conference on Machine
  Learning (ICML)}, 2017.

\bibitem[BKL{\etalchar{+}}19]{bkllsw19}
Fernando~GSL Brand{\~a}o, Amir Kalev, Tongyang Li, Cedric Yen-Yu Lin, Krysta~M
  Svore, and Xiaodi Wu.
\newblock Quantum sdp solvers: Large speed-ups, optimality, and applications to
  quantum learning.
\newblock In {\em 46th International Colloquium on Automata, Languages, and
  Programming (ICALP)}. Schloss Dagstuhl-Leibniz-Zentrum fuer Informatik, 2019.

\bibitem[BLSS20]{blss20}
Jan van~den Brand, Yin~Tat Lee, Aaron Sidford, and Zhao Song.
\newblock Solving tall dense linear programs in nearly linear time.
\newblock In {\em 52nd Annual ACM SIGACT Symposium on Theory of Computing
  (STOC)}, 2020.

\bibitem[Bra20]{b20}
Jan van~den Brand.
\newblock A deterministic linear program solver in current matrix
  multiplication time.
\newblock In {\em ACM-SIAM Symposium on Discrete Algorithms (SODA)}, 2020.

\bibitem[BS17]{bs17}
Fernando~GSL Brandao and Krysta~M Svore.
\newblock Quantum speed-ups for solving semidefinite programs.
\newblock In {\em 2017 IEEE 58th Annual Symposium on Foundations of Computer
  Science (FOCS)}, pages 415--426. IEEE, 2017.

\bibitem[BTN01]{bn01}
Aharon Ben-Tal and Arkadi Nemirovski.
\newblock {\em Lectures on modern convex optimization: analysis, algorithms,
  and engineering applications}.
\newblock SIAM, 2001.

\bibitem[Bub15]{b15}
S{\'e}bastien Bubeck.
\newblock Convex optimization: Algorithms and complexity.
\newblock {\em Foundations and Trends{\textregistered} in Machine Learning},
  8(3-4):231--357, 2015.

\bibitem[CCH{\etalchar{+}}22]{cch+22}
Nadiia Chepurko, Kenneth Clarkson, Lior Horesh, Honghao Lin, and David
  Woodruff.
\newblock Quantum-inspired algorithms from randomized numerical linear algebra.
\newblock In {\em International Conference on Machine Learning}, pages
  3879--3900. PMLR, 2022.

\bibitem[CDHS19]{cdhs19}
Yair Carmon, John~C Duchi, Oliver Hinder, and Aaron Sidford.
\newblock Lower bounds for finding stationary points i.
\newblock {\em Mathematical Programming}, pages 1--50, 2019.

\bibitem[CDST19]{cdst19}
Yair Carmon, John~C. Duchi, Aaron Sidford, and Kevin Tian.
\newblock A rank-1 sketch for matrix multiplicative weights.
\newblock In {\em Conference on Learning Theory (COLT)}, pages 589--623, 2019.

\bibitem[CGJ18]{cgj18}
Shantanav Chakraborty, Andr{\'a}s Gily{\'e}n, and Stacey Jeffery.
\newblock The power of block-encoded matrix powers: improved regression
  techniques via faster hamiltonian simulation.
\newblock {\em arXiv preprint arXiv:1804.01973}, 2018.

\bibitem[CGL{\etalchar{+}}20]{cgl+20}
Nai-Hui Chia, Andr{\'a}s Gily{\'e}n, Tongyang Li, Han-Hsuan Lin, Ewin Tang, and
  Chunhao Wang.
\newblock Sampling-based sublinear low-rank matrix arithmetic framework for
  dequantizing quantum machine learning.
\newblock In {\em Proceedings of the 52nd Annual ACM SIGACT symposium on theory
  of computing (STOC)}, pages 387--400, 2020.

\bibitem[CLS19]{cls19}
Michael~B Cohen, Yin~Tat Lee, and Zhao Song.
\newblock Solving linear programs in the current matrix multiplication time.
\newblock In {\em Proceedings of the 51st Annual ACM Symposium on Theory of
  Computing (STOC)}, 2019.

\bibitem[CMD20]{cm20}
Pablo~AM Casares and Miguel~Angel Martin-Delgado.
\newblock A quantum interior-point predictor--corrector algorithm for linear
  programming.
\newblock {\em Journal of physics A: Mathematical and Theoretical},
  53(44):445305, 2020.

\bibitem[CVB20]{cv20}
Daan Camps and Roel Van~Beeumen.
\newblock Approximate quantum circuit synthesis using block encodings.
\newblock {\em Physical Review A}, 102(5):052411, 2020.

\bibitem[Dan47]{d47}
George~B Dantzig.
\newblock Maximization of a linear function of variables subject to linear
  inequalities.
\newblock {\em Activity analysis of production and allocation}, 13:339--347,
  1947.

\bibitem[DLY21]{dly21}
Sally Dong, Yin~Tat Lee, and Guanghao Ye.
\newblock A nearly-linear time algorithm for linear programs with small
  treewidth: A multiscale representation of robust central path.
\newblock In {\em Proceedings of the 53rd Annual ACM SIGACT Symposium on Theory
  of Computing (STOC)}, 2021.

\bibitem[GH16]{gh16}
Dan Garber and Elad Hazan.
\newblock Sublinear time algorithms for approximate semidefinite programming.
\newblock {\em Mathematical Programming}, 158(1-2):329--361, 2016.

\bibitem[GLM08]{gio08}
Vittorio Giovannetti, Seth Lloyd, and Lorenzo Maccone.
\newblock Quantum random access memory.
\newblock {\em Physical review letters}, 100(16):160501, 2008.

\bibitem[GSLW19]{gslw19}
Andr{\'a}s Gily{\'e}n, Yuan Su, Guang~Hao Low, and Nathan Wiebe.
\newblock Quantum singular value transformation and beyond: exponential
  improvements for quantum matrix arithmetics.
\newblock In {\em Proceedings of the 51st Annual ACM SIGACT Symposium on Theory
  of Computing}, pages 193--204, 2019.

\bibitem[GST22]{gst22}
Andr{\'a}s Gily{\'e}n, Zhao Song, and Ewin Tang.
\newblock An improved quantum-inspired algorithm for linear regression.
\newblock {\em Quantum}, 6:754, 2022.

\bibitem[HJS{\etalchar{+}}22]{hjstz22}
Baihe Huang, Shunhua Jiang, Zhao Song, Runzhou Tao, and Ruizhe Zhang.
\newblock Solving sdp faster: A robust ipm framework and efficient
  implementation.
\newblock In {\em FOCS}, 2022.

\bibitem[JKL{\etalchar{+}}20]{jklps20}
Haotian Jiang, Tarun Kathuria, Yin~Tat Lee, Swati Padmanabhan, and Zhao Song.
\newblock A faster interior point method for semidefinite programming.
\newblock In {\em FOCS}, 2020.

\bibitem[JLSW20]{jlsw20}
Haotian Jiang, Yin~Tat Lee, Zhao Song, and Sam Chiu-wai Wong.
\newblock An improved cutting plane method for convex optimization,
  convex-concave games and its applications.
\newblock In {\em STOC}, 2020.

\bibitem[JSWZ21]{jswz21}
Shunhua Jiang, Zhao Song, Omri Weinstein, and Hengjie Zhang.
\newblock Faster dynamic matrix inverse for faster lps.
\newblock In {\em STOC}, 2021.

\bibitem[JY11]{jy11}
Rahul Jain and Penghui Yao.
\newblock A parallel approximation algorithm for positive semidefinite
  programming.
\newblock In {\em Proceedings of the 2011 IEEE 52nd Annual Symposium on
  Foundations of Computer Science (FOCS)}, 2011.

\bibitem[Kar84]{k84}
Narendra Karmarkar.
\newblock A new polynomial-time algorithm for linear programming.
\newblock In {\em Proceedings of the sixteenth annual ACM symposium on Theory
  of computing (STOC)}, 1984.

\bibitem[Kha80]{k80}
Leonid~G Khachiyan.
\newblock Polynomial algorithms in linear programming.
\newblock {\em USSR Computational Mathematics and Mathematical Physics},
  20(1):53--72, 1980.

\bibitem[KLP19]{klp19}
Iordanis Kerenidis, Jonas Landman, and Anupam Prakash.
\newblock Quantum algorithms for deep convolutional neural networks.
\newblock {\em arXiv preprint arXiv:1911.01117}, 2019.

\bibitem[KM03]{km03}
Kartik Krishnan and John~E Mitchell.
\newblock Properties of a cutting plane method for semidefinite programming.
\newblock {\em submitted for publication}, 2003.

\bibitem[KP17]{kp17}
Iordanis Kerenidis and Anupam Prakash.
\newblock Quantum recommendation systems.
\newblock In {\em 8th Innovations in Theoretical Computer Science Conference
  (ITCS)}. Schloss Dagstuhl-Leibniz-Zentrum fuer Informatik, 2017.

\bibitem[KP20]{kp20}
Iordanis Kerenidis and Anupam Prakash.
\newblock A quantum interior point method for lps and sdps.
\newblock {\em ACM Transactions on Quantum Computing}, 1(1):1--32, 2020.

\bibitem[KPS21]{kps21}
Iordanis Kerenidis, Anupam Prakash, and D{\'a}niel Szil{\'a}gyi.
\newblock Quantum algorithms for second-order cone programming and support
  vector machines.
\newblock {\em Quantum}, 5:427, 2021.

\bibitem[KTE88]{kte88}
Leonid~G Khachiyan, Sergei~Pavlovich Tarasov, and I.~I. Erlikh.
\newblock The method of inscribed ellipsoids.
\newblock In {\em Soviet Math. Dokl}, volume~37, pages 226--230, 1988.

\bibitem[LC19]{lc19}
Guang~Hao Low and Isaac~L Chuang.
\newblock Hamiltonian simulation by qubitization.
\newblock {\em Quantum}, 3:163, 2019.

\bibitem[LP20]{lp19}
Yin~Tat Lee and Swati Padmanabhan.
\newblock An $\widetilde{O}(m/\epsilon^{3.5})$-cost algorithm for semidefinite
  programs with diagonal constraints.
\newblock In {\em Conference on Learning Theory (COLT)}, Proceedings of Machine
  Learning Research. {PMLR}, 2020.

\bibitem[LS14]{ls14}
Yin~Tat Lee and Aaron Sidford.
\newblock Path finding methods for linear programming: Solving linear programs
  in ${O}( \sqrt{rank} )$ iterations and faster algorithms for maximum flow.
\newblock In {\em 55th Annual IEEE Symposium on Foundations of Computer Science
  (FOCS)}, 2014.

\bibitem[LS15]{ls15}
Yin~Tat Lee and Aaron Sidford.
\newblock Efficient inverse maintenance and faster algorithms for linear
  programming.
\newblock In {\em 56th Annual IEEE Symposium on Foundations of Computer Science
  (FOCS)}, 2015.

\bibitem[LSW15]{lsw15}
Yin~Tat Lee, Aaron Sidford, and Sam~Chiu{-}wai Wong.
\newblock A faster cutting plane method and its implications for combinatorial
  and convex optimization.
\newblock In {\em 56th Annual IEEE Symposium on Foundations of Computer Science
  (FOCS)}, 2015.

\bibitem[LSZ19]{lsz19}
Yin~Tat Lee, Zhao Song, and Qiuyi Zhang.
\newblock Solving empirical risk minimization in the current matrix
  multiplication time.
\newblock In {\em Annual Conference on Learning Theory (COLT)}, 2019.

\bibitem[MT00]{mt00}
Renato~DC Monteiro and Takashi Tsuchiya.
\newblock Polynomial convergence of primal-dual algorithms for the second-order
  cone program based on the mz-family of directions.
\newblock {\em Mathematical programming}, 88(1):61--83, 2000.

\bibitem[MZ10]{mz10}
Lingsheng Meng and Bing Zheng.
\newblock The optimal perturbation bounds of the moore--penrose inverse under
  the frobenius norm.
\newblock {\em Linear algebra and its applications}, 432(4):956--963, 2010.

\bibitem[NC11]{nc11}
Michael~A. Nielsen and Isaac~L. Chuang.
\newblock {\em Quantum Computation and Quantum Information: 10th Anniversary
  Edition}.
\newblock Cambridge University Press, USA, 10th edition, 2011.

\bibitem[NN89]{nn89}
Yurii Nesterov and Arkadi Nemirovski.
\newblock Self-concordant functions and polynomial time methods in convex
  programming. preprint, central economic \& mathematical institute, ussr acad.
\newblock {\em Sci. Moscow, USSR}, 1989.

\bibitem[NN92]{nn92}
Yurii Nesterov and Arkadi Nemirovski.
\newblock Conic formulation of a convex programming problem and duality.
\newblock {\em Optimization Methods and Software}, 1(2):95--115, 1992.

\bibitem[NN94]{nn94}
Yurii Nesterov and Arkadi Nemirovski.
\newblock {\em Interior-point polynomial algorithms in convex programming},
  volume~13.
\newblock Siam, 1994.

\bibitem[Ren88]{r88}
James Renegar.
\newblock A polynomial-time algorithm, based on newton's method, for linear
  programming.
\newblock {\em Mathematical Programming}, 40(1-3):59--93, 1988.

\bibitem[Ren01]{r01}
James Renegar.
\newblock {\em A Mathematical View of Interior-point Methods in Convex
  Optimization}.
\newblock Society for Industrial and Applied Mathematics, Philadelphia, PA,
  USA, 2001.

\bibitem[Sho77]{s77}
Naum~Z Shor.
\newblock Cut-off method with space extension in convex programming problems.
\newblock {\em Cybernetics and systems analysis}, 13(1):94--96, 1977.

\bibitem[SY21]{sy21}
Zhao Song and Zheng Yu.
\newblock Oblivious sketching-based central path method for solving linear
  programming.
\newblock In {\em ICML}, 2021.

\bibitem[SYZ21]{syz21}
Zhao Song, Shuo Yang, and Ruizhe Zhang.
\newblock Does preprocessing help training over-parameterized neural networks?
\newblock In {\em Thirty-Fifth Conference on Neural Information Processing
  Systems (NeurIPS)}, 2021.

\bibitem[Tan19]{t19}
Ewin Tang.
\newblock A quantum-inspired classical algorithm for recommendation systems.
\newblock In {\em Proceedings of the 51st Annual ACM SIGACT Symposium on Theory
  of Computing}, pages 217--228, 2019.

\bibitem[Tan21]{t21}
Ewin Tang.
\newblock Quantum principal component analysis only achieves an exponential
  speedup because of its state preparation assumptions.
\newblock {\em Physical Review Letters}, 127(6):060503, 2021.

\bibitem[vACGN22]{vcg22}
Joran van Apeldoorn, Arjan Cornelissen, Andr{\'a}s Gily{\'e}n, and Giacomo
  Nannicini.
\newblock Quantum tomography using state-preparation unitaries.
\newblock {\em arXiv preprint arXiv:2207.08800}, 2022.

\bibitem[Vai89a]{v89}
Pravin~M Vaidya.
\newblock A new algorithm for minimizing convex functions over convex sets.
\newblock In {\em 30th Annual IEEE Symposium on Foundations of Computer Science
  (FOCS)}, pages 338--343, 1989.

\bibitem[Vai89b]{v89_lp}
Pravin~M Vaidya.
\newblock Speeding-up linear programming using fast matrix multiplication.
\newblock In {\em 30th Annual Symposium on Foundations of Computer Science
  (FOCS)}, pages 332--337. IEEE, 1989.

\bibitem[Wed73]{w73}
Per-{\AA}ke Wedin.
\newblock Perturbation theory for pseudo-inverses.
\newblock {\em BIT Numerical Mathematics}, 13(2):217--232, 1973.

\bibitem[Ye21]{y21}
Guanghao Ye.
\newblock Fast algorithm for solving structured convex programs.
\newblock Master's thesis, The University of Washington, 2021.

\bibitem[YN76]{yn76}
David~B Yudin and Arkadi~S Nemirovski.
\newblock Evaluation of the information complexity of mathematical programming
  problems.
\newblock {\em Ekonomika i Matematicheskie Metody}, 12:128--142, 1976.

\bibitem[YTF{\etalchar{+}}19]{ytfuc19}
Alp Yurtsever, Joel~A. Tropp, Olivier Fercoq, Madeleine Udell, and Volkan
  Cevher.
\newblock Scalable semidefinite programming, 2019.

\end{thebibliography}
\appendix
\section{Well-conditioned SDP Instances}

Here we list some SDP cases in high dimensions where the condition number of Hessian does not depend on $1/\epsilon$. 
\paragraph{Case 1: }A simple case is where the central path is \emph{nearly isotropic} and thus the condition number of Hessian is close to 1, e.g.
\begin{align*}
    \max_{y\in \R^m} ~y_1 + \dots +  y_m ~~\text{subject to}~~ y_i \leq 0~~ \forall i \in [m].
\end{align*}
\paragraph{Case 2: }Another SDP example is as follows: 
\begin{align*}
\max_{X \in \R^{n \times n}} ~ \langle C, X \rangle %\notag \\
\textup{ subject to } ~ \langle A_i, X \rangle = b_i, ~~\forall i \in [m], %\\
~ X \succeq 0,% \notag
\end{align*}
where $n = m = 3$ and
\begin{align*}
    C = \begin{bmatrix}0 & 0 & 0\\
    0 & 0 & 0 \\
    0 & 0 & -1\end{bmatrix},
    A_1 = \begin{bmatrix}-1 & 0 & 0\\
    0 & 0 & 0 \\
    0 & 0 & 0\end{bmatrix},
    A_2 = \begin{bmatrix}0 & 0 & 0\\
    0 & -1 & 0 \\
    0 & 0 & 0\end{bmatrix},
    A_3 = \begin{bmatrix}0 & -1 & 0\\
    -1 & 0 & 0 \\
    0 & 0 & 0\end{bmatrix}, b = \begin{bmatrix}-1\\-1\\0\end{bmatrix}.
\end{align*}
The dual problem is given by
\begin{align*}
    \min_{y \in \R^m} ~ b^\top y ~~~ \mathrm{~subject~to} ~ S = \sum_{i=1}^m y_i A_i - C, 
    ~~~~ S \succeq 0. 
\end{align*}
The central path can be found by solving the following penalized objective function
\begin{align*}
    \min_{y \in \R^m} ~ \eta \cdot b^\top y - \log \det (S)
\end{align*}
which yields
\begin{align*}
    y = \begin{bmatrix}-\eta^{-1} \\-\eta^{-1}\\0\end{bmatrix}, S = \begin{bmatrix}\eta^{-1} & 0 & 0\\
    0 & \eta^{-1} & 0 \\
    0 & 0 & 1\end{bmatrix}.
\end{align*}
Therefore the Hessian is given by
\begin{align*}
    H = \begin{bmatrix}\eta^2 & 0 & 0\\
    0 & \eta^2 & 0 \\
    0 & 0 & 2\eta^2\end{bmatrix}
\end{align*}
and the condition number is $\kappa(H) = 2$. Notice that the primal and dual solutions are
\begin{align*}
    X = \begin{bmatrix}1 & 0 & 0\\
    0 & 1 & 0 \\
    0 & 0 & 0\end{bmatrix}, y = \begin{bmatrix}0 \\0\\0\end{bmatrix}, S = \begin{bmatrix}0 & 0 & 0\\
    0 & 0 & 0 \\
    0 & 0 & 1\end{bmatrix},
\end{align*}
thus the strict complementary holds. In addition, another intermediate-matrix related term in the running time $\|\mathsf{A}^\top y\|_2^2  \|S^{-1}\|^2=2$ in this example. Notice also the strict complementarity holds. Therefore, our algorithm will have $\log(1/\epsilon)$ dependence in the total running time in this example. In brief, the “good cases” (where our quantum algorithm’s running time depends on $\log(1/\epsilon)$) for SDP are much more non-trivial than those for LP.

We also note that in all previous quantum interior-point method papers (\cite{kp20,cm20,kps21,antz21}), their running times also depend on the condition number of the linear system that appears in each iteration. In addition, their running times explicitly depend on $\poly(1/\epsilon)$ as well. Thus, even in the well-conditioned cases, their running times still depend on $\poly(1/\epsilon)$. And within the quantum linear algebra framework, almost all matrix-related algorithms depend on the condition number of the matrix. Hence, it seems that the dependence on the condition number is unavoidable for the current techniques in quantum computing. 

\end{document}